\documentclass[letterpaper,11pt]{article}
\usepackage[utf8]{inputenc}
\pdfoutput=1
\usepackage{jheppub}
\usepackage{mathtools,amssymb,braket,mathrsfs,tikz,subfig,xspace,xcolor,float,color,siunitx,comment,afterpage,multirow,array,hhline,amsmath,amsthm,framed,colortbl}
\usepackage[countmax]{subfloat}

\providecommand{\href}[2]{#2}

\DeclareMathOperator*{\argmin}{arg\,min}
\newcommand{\E}{\mathcal{E}}
\renewcommand{\O}{\mathcal{O}}

\DeclareRobustCommand{\Sec}[1]{Sec.~\ref{#1}}
\DeclareRobustCommand{\Secs}[2]{Secs.~\ref{#1} and \ref{#2}}

\DeclareRobustCommand{\App}[1]{App.~\ref{#1}}
\DeclareRobustCommand{\Tab}[1]{Table~\ref{#1}}

\DeclareRobustCommand{\Fig}[1]{Fig.~\ref{#1}}

\DeclareRobustCommand{\Eq}[1]{Eq.~(\ref{#1})}
\DeclareRobustCommand{\Eqs}[2]{Eqs.~(\ref{#1}) and (\ref{#2})}

\DeclareRobustCommand{\Ref}[1]{Ref.~\cite{#1}}
\DeclareRobustCommand{\Refs}[1]{Refs.~\cite{#1}}

\newtheorem{lemma}{Lemma}

\newtheorem{definition}{Definition}
\newtheorem{conjecture}{Conjecture}
\newtheorem*{ircsafety}{Infrared and Collinear Safety}

\newcommand{\EMD}{\text{EMD}\xspace}

\definecolor{jdtcolor}{rgb}{0.8,0,0}
\definecolor{emmcolor}{rgb}{0,0.8,0}
\definecolor{ptkcolor}{rgb}{0,0,0.8}

\usepackage{amsthm}

\definecolor{table_red}{rgb}{0.501960,0.000000,0.008131}
\definecolor{table_yellow}{rgb}{0.501959,0.501930,0.014756}
\definecolor{table_green}{rgb}{0.065930,0.501799,0.006832}
\definecolor{table_teal}{rgb}{0.064121,0.501820,0.501973}

\definecolor{table_red_bg}{RGB}{249,243,242}
\definecolor{table_yellow_bg}{RGB}{249,249,243}
\definecolor{table_green_bg}{RGB}{245,249,244}
\definecolor{table_teal_bg}{RGB}{244,249,249}

\title{The Hidden Geometry of Particle Collisions}

\preprint{MIT--CTP 5185}

\author{Patrick T. Komiske,}
\author{Eric M. Metodiev,}
\author{and Jesse Thaler}

\affiliation{Center for Theoretical Physics, Massachusetts Institute of Technology, Cambridge, MA 02139, USA}

\emailAdd{pkomiske@mit.edu}
\emailAdd{metodiev@mit.edu}
\emailAdd{jthaler@mit.edu}

\abstract{
We establish that many fundamental concepts and techniques in quantum field theory and collider physics can be naturally understood and unified through a simple new geometric language.
The idea is to equip the space of collider events with a metric, from which other geometric objects can be rigorously defined.
Our analysis is based on the energy mover's distance, which quantifies the ``work'' required to rearrange one event into another.
This metric, which operates purely at the level of observable energy flow information, allows for a clarified definition of infrared and collinear safety and related concepts.
A number of well-known collider observables can be exactly cast as the minimum distance between an event and various manifolds in this space.
Jet definitions, such as exclusive cone and sequential recombination algorithms, can be directly derived by finding the closest few-particle approximation to the event.
Several area- and constituent-based pileup mitigation strategies are naturally expressed in this formalism as well.
Finally, we lift our reasoning to develop a precise distance between theories, which are treated as collections of events weighted by cross sections.
In all of these various cases, a better understanding of existing methods in our geometric language suggests interesting new ideas and generalizations.
}

\begin{document} 

\flushbottom
\maketitle

\addtocontents{toc}{\protect\enlargethispage{9mm}}

\newpage

\section{Introduction}
\label{sec:intro}

Unification of ideas in physics has been an important way of achieving elegance, clarity, and simplicity, which in turn helps inspire meaningful new developments.
In this paper, we use the energy mover's distance (EMD) between collider events~\cite{Komiske:2019fks} to provide a natural geometric language that unifies many important concepts and techniques in quantum field theory and collider physics from the past five decades.
Furthermore, we introduce and discuss several new ideas inspired by this geometric approach to studying the space of events.

Throughout this paper, we refer to an event and its \emph{energy flow} interchangeably.
The energy flow, or distribution of energy, is the kinematic information that is experimentally observable and perturbatively well-defined in quantum field theories with massless particles~\cite{Tkachov:1995kk}.
As it relates to collider physics, the energy flow has been extensively studied~\cite{Tkachov:1995kk,Sveshnikov:1995vi,Korchemsky:1997sy,Basham:1978zq,Cherzor:1997ak,Tkachov:1999py,Korchemsky:1999kt,Belitsky:2001ij,Berger:2002jt,Bauer:2008dt,Hofman:2008ar,Mateu:2012nk,Belitsky:2013xxa,Komiske:2017aww,Komiske:2018cqr,Komiske:2019asc}, and this paper builds on many of these previous concepts.
For an event consisting of $M$ particles with positive energies $E_i$ and angular directions $\hat n_i$, the energy flow is:
\begin{equation}
\label{eq:energyflow}
\E(\hat n) = \sum_{i=1}^ME_i\,\delta(\hat n - \hat n_i).
\end{equation}
Note that the energy flow is insensitive to charge and flavor information.
Particles are taken to be massless in the body of this paper, with $n_i^\mu = (1, \hat{n}_i)^\mu = p^\mu_i/E_i$, and the case of massive particles is discussed in \App{sec:mass}.
In a hadron collider context, particle transverse momenta $p_{T,i}$ are typically used in place of particle energies, but we focus on energies in this paper to minimize extraneous notation.

The EMD was introduced in \Ref{Komiske:2019fks} as a metric between events.
It is based on the well-known earth mover's distance~\cite{DBLP:journals/pami/PelegWR89,Rubner:1998:MDA:938978.939133,Rubner:2000:EMD:365875.365881,DBLP:conf/eccv/PeleW08,DBLP:conf/gsi/PeleT13}, also known as the Wasserstein metric~\cite{wasserstein1969markov,dobrushin1970prescribing}.
Intuitively, the EMD between two events is the amount of ``work'' required to rearrange one event to the other.
Its value can be obtained by solving the following optimal transport problem between energy flows $\E$ and $\E'$:
\begin{equation}
\label{eq:emd}
\EMD_{\beta,R} (\mathcal E, \mathcal E') = \min_{\{f_{ij}\ge0\}} \sum_{i=1}^M\sum_{j=1}^{M'} f_{ij} \left( \frac{\theta_{ij}}{R} \right)^\beta + \left|\sum_{i=1}^M E_i - \sum_{j=1}^{M'}E_j'\right|,
\end{equation}
\begin{equation}
\label{eq:emdconstraints}
\sum_{i=1}^M f_{ij} \le E_j', \quad\quad \sum_{j=1}^{M'} f_{ij} \le E_i, \quad\quad \sum_{i=1}^M\sum_{j=1}^{M'} f_{ij} = \min\left(\sum_{i=1}^M E_i,\sum_{j=1}^{M'}E_j'\right),
\end{equation}
where $\theta_{ij}$ is a pairwise distance between particles known as the \emph{ground metric}, $R>0$ is a parameter controlling the tradeoff between transporting energy and destroying it, and $\beta > 0$ is an angular weighting exponent.%
\footnote{\label{footnote:pWasser}
Strictly speaking, for the case of $\beta>1$, one must raise the first term in \Eq{eq:emd} to the $1/\beta$ power for the EMD to be a proper metric satisfying the triangle inequality, in which case it is known as a $p$-Wasserstein metric with $p = \beta$.
Additionally, $2R$ should be larger than or equal to the maximum distance in the ground space for the EMD to satisfy the triangle inequality.
When written without subscripts, $\EMD(\E,\E')$ refers to the case of $\beta=1$ and a sufficiently large $R$ to ensure that we have a proper metric.
Even if the EMD is not a proper metric, though, it is still a valid optimal transport problem for any positive values of $\beta$ and $R$.
}
For the angular metric between two massless particles, we focus on the case of
\begin{equation}
\label{eq:theta_def}
\theta_{ij} = \sqrt{2n_i^\mu n_{j\mu}} = \sqrt{2 (1 - \hat{n}_i \cdot \hat{n}_j)},
\end{equation}
which reduces to their opening angle in the nearby limit.%
\footnote{Many modifications to this EMD definition are possible, including alternative angular distances such as strict opening angle or rapidity-azimuth distance as well as alternative notions of energy such as transverse momentum. In addition, energies can be normalized by dividing by their total scalar sum so that energy flows become proper probability distributions.  If desired, the EMD in the center-of-mass frame can be phrased in a manifestly Lorentz-invariant way by replacing the particle energies $E_i$ with $p_i^\mu P_\mu/\sqrt{P_\mu P^\mu}$, where $P_\mu$ is the total event four-momentum.}
The first term in \Eq{eq:emd} quantifies the difference in radiation patterns while the second term, which vanishes in the case of normalized energy flows, allows for the comparison of events with different total energies.
The constraints in \Eq{eq:emdconstraints} specify that the amount of energy moved to or from a particle cannot exceed its initial energy, and that as much energy must be moved as possible.

\begin{figure}[t]
\centering
\includegraphics[scale=0.7]{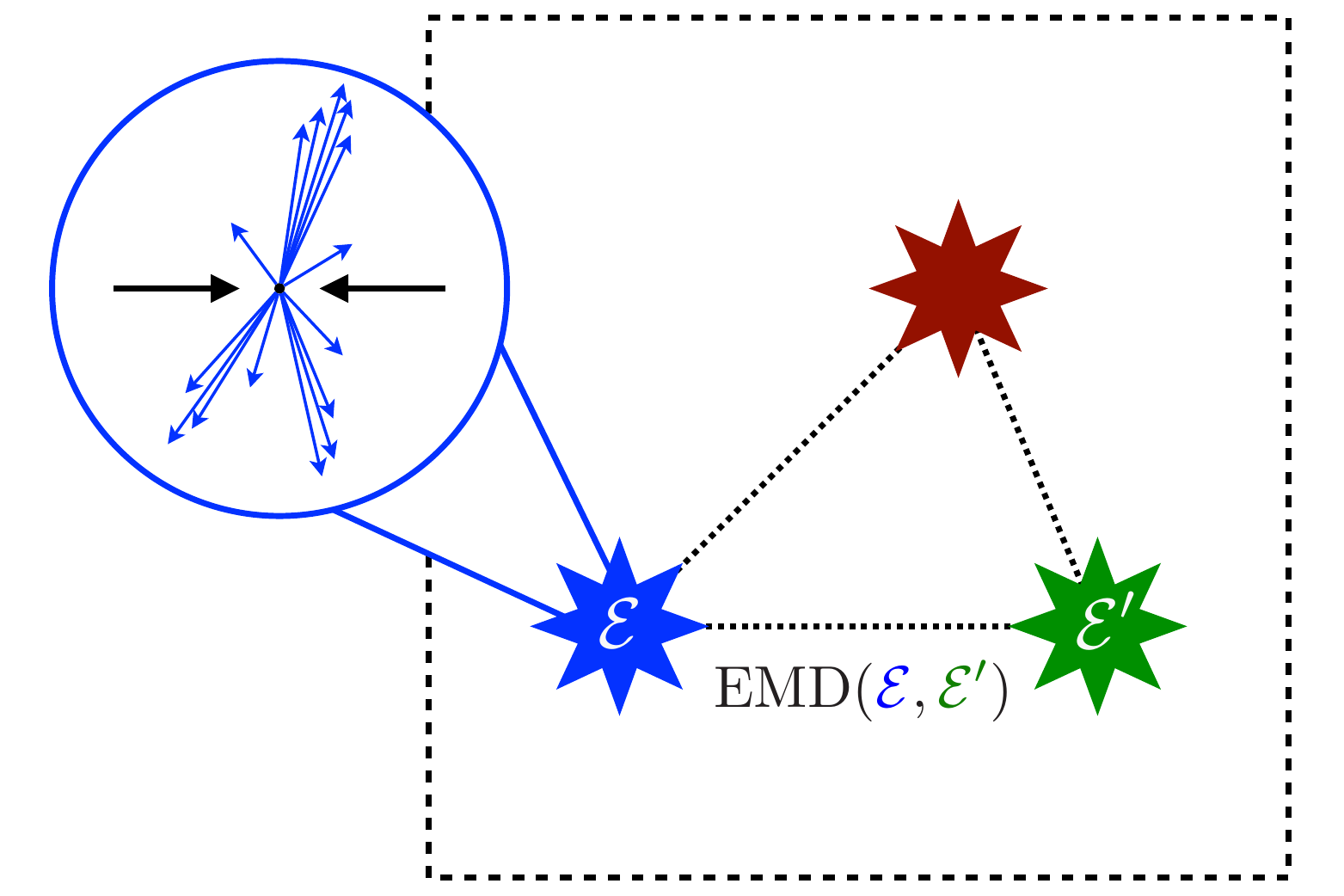}
\caption{\label{fig:event_space} An illustration of the space of events.
Each point in the space is a collider event consisting of the particles produced in a collision, as indicated by the blue event.
The distance between events is quantified by the EMD, giving rise to a metric space.
Geometry in this abstract space of events provides a natural language to understand many ideas and developments in quantum field theory and collider physics.}
\end{figure}

The EMD has previously been used to bound modifications to infrared- and collinear-safe (IRC-safe) observables, distinguish different types of jets, and enable visualizations of the space of events~\cite{Komiske:2019fks}.
It has also been used to explore the space of jets and quantify detector effects with CMS Open Data from the Large Hadron Collider (LHC)~\cite{Komiske:2019jim}.
Alternative pairwise event distances were considered in \Ref{Mullin:2019mmh} in the context of new physics searches.
Here, we demonstrate that the EMD can be used to clarify numerous concepts throughout quantum field theory and collider physics using a unified language of event space geometry.
In addition to demonstrating how concepts such as IRC safety, observables, jet finding, and pileup subtraction are related, we will develop new ideas and techniques in each of these areas, which we describe below.

Equipping collider events with a metric allows us to explore interesting geometric and topological ideas in the space of events.
\Fig{fig:event_space} illustrates the space of events with the EMD as a metric.
One key construction for relating these concepts is the notion of a manifold in the space of events, which will allow us to define the distance between an event and a manifold, as well as the point of closest approach on a manifold.
Since fixed-order perturbation theory works with a definite number of particles, an important type of manifold will be the idealized massless $N$-particle manifold:
\begin{equation}
\label{eq:npmanifold}
\mathcal P_N = \left\{\left.\sum_{i=1}^N E_i\, \delta(\hat n - \hat n_i)\,\, \right| \,\, E_i\ge0 \right\},
\end{equation}
which, intuitively, is the set of all possible events with $N$ massless particles.
Note that $\mathcal P_N\supset\mathcal P_{N-1}\supset\cdots\mathcal P_2\supset \mathcal P_1\supset\mathcal P_0$ via soft and collinear limits, so that the idealized $N$-particle manifold contains each manifold of smaller particle multiplicity.

\begin{table}[t]
\centering
\begin{tabular}{|c|c|c|c|}
\hline\hline
\bf Sec. & \bf Concept &  \bf Equation & \bf Illustration  \\
\hline\hline
& &  &  \multirow{5}{*}{\raisebox{-7em}{\includegraphics[scale=0.28]{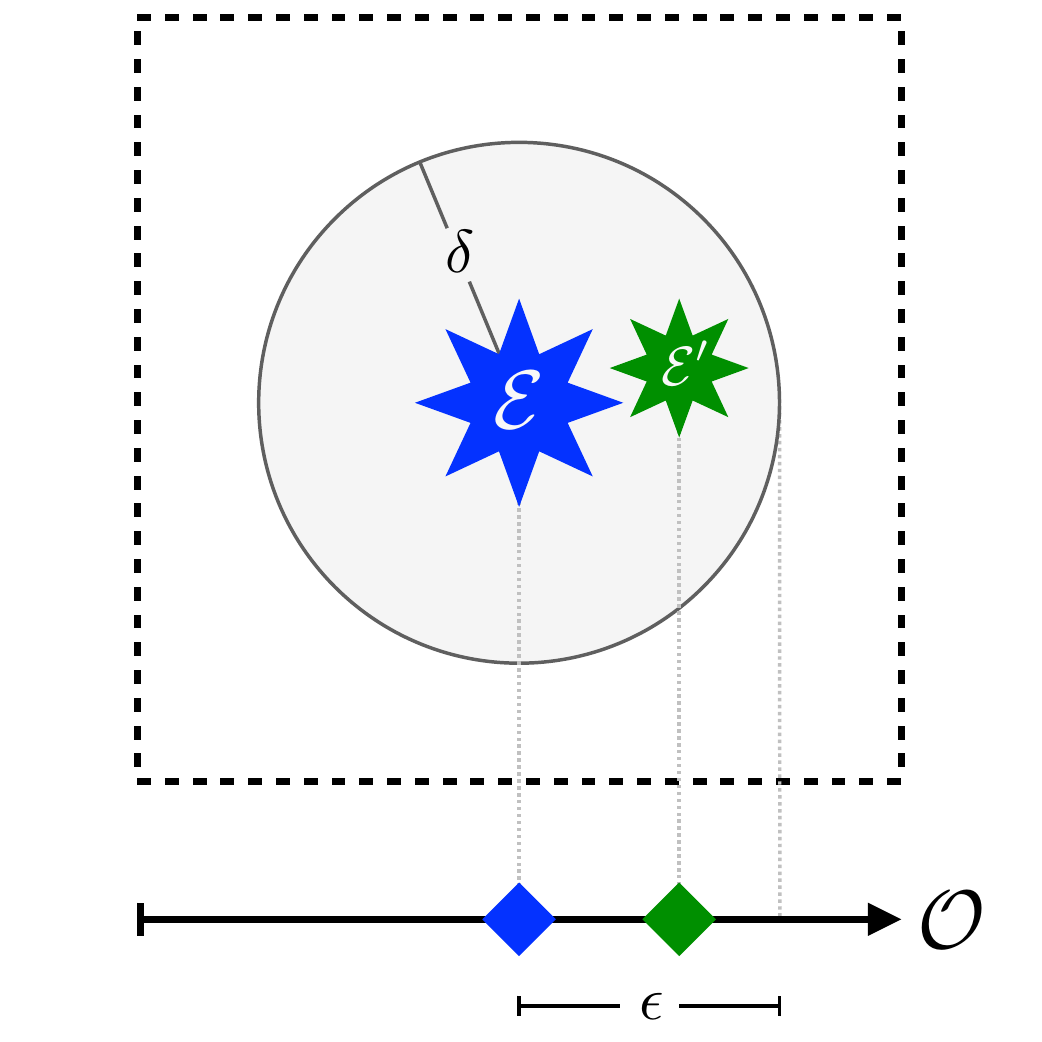}}} \\
 \ref{sec:safety} & {\bf Infrared and} & $\displaystyle\text{EMD}(\E,\E')<\delta  \implies$ & \\ 
 & {\bf Collinear Safety} & $|\O(\E) - \O(\E')| < \epsilon$ &  \\
& \cite{Kinoshita:1962ur,Lee:1964is,Sterman:1977wj,Sterman:1978bi,Sterman:1978bj,Sterman:1979uw,sterman1995handbook,Weinberg:1995mt,Ellis:1991qj,Banfi:2004yd,Larkoski:2013paa,Larkoski:2014wba,Larkoski:2015lea} &  & \\
& &  &  \\ \hline
& &  & \multirow{5}{*}{\raisebox{-7em}{\includegraphics[scale=0.28]{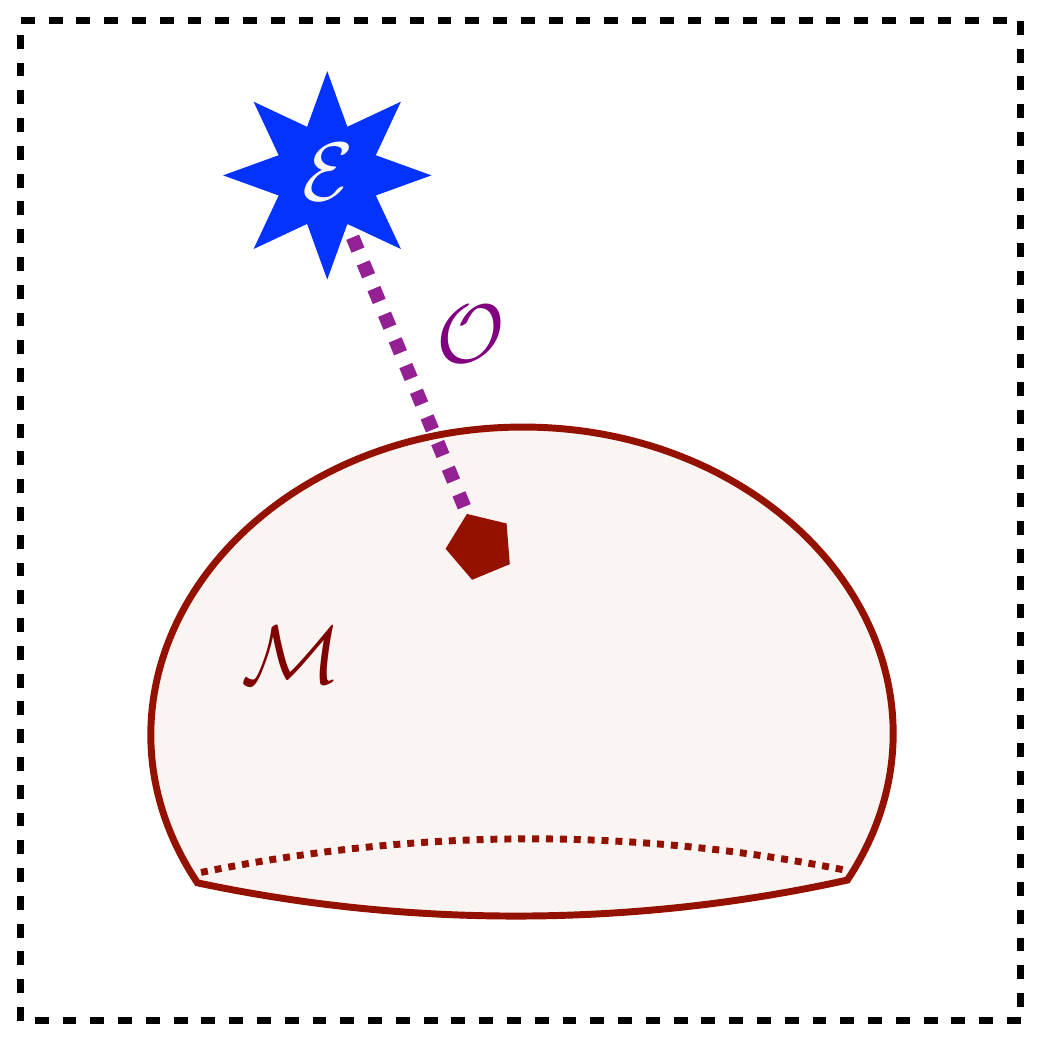}}} \\ 
 \ref{sec:observables} &{\bf Observables} &  $\displaystyle\O(\E) = \min_{\E'\in\mathcal M}\EMD(\E,\E')$ &   \\
\ref{sec:eventobservables} & Event Shapes~\cite{Brandt:1964sa,Farhi:1977sg,Georgi:1977sf,Larkoski:2014uqa,Stewart:2010tn,isotropytemp} &  &    \\
 \ref{sec:jetobservables} & Jet Shapes~\cite{Ellis:2010rwa,Thaler:2010tr,Thaler:2011gf}  &  &    \\
& &  &  \\ \hline
& &  & \multirow{5}{*}{\raisebox{-7em}{\includegraphics[scale=0.28]{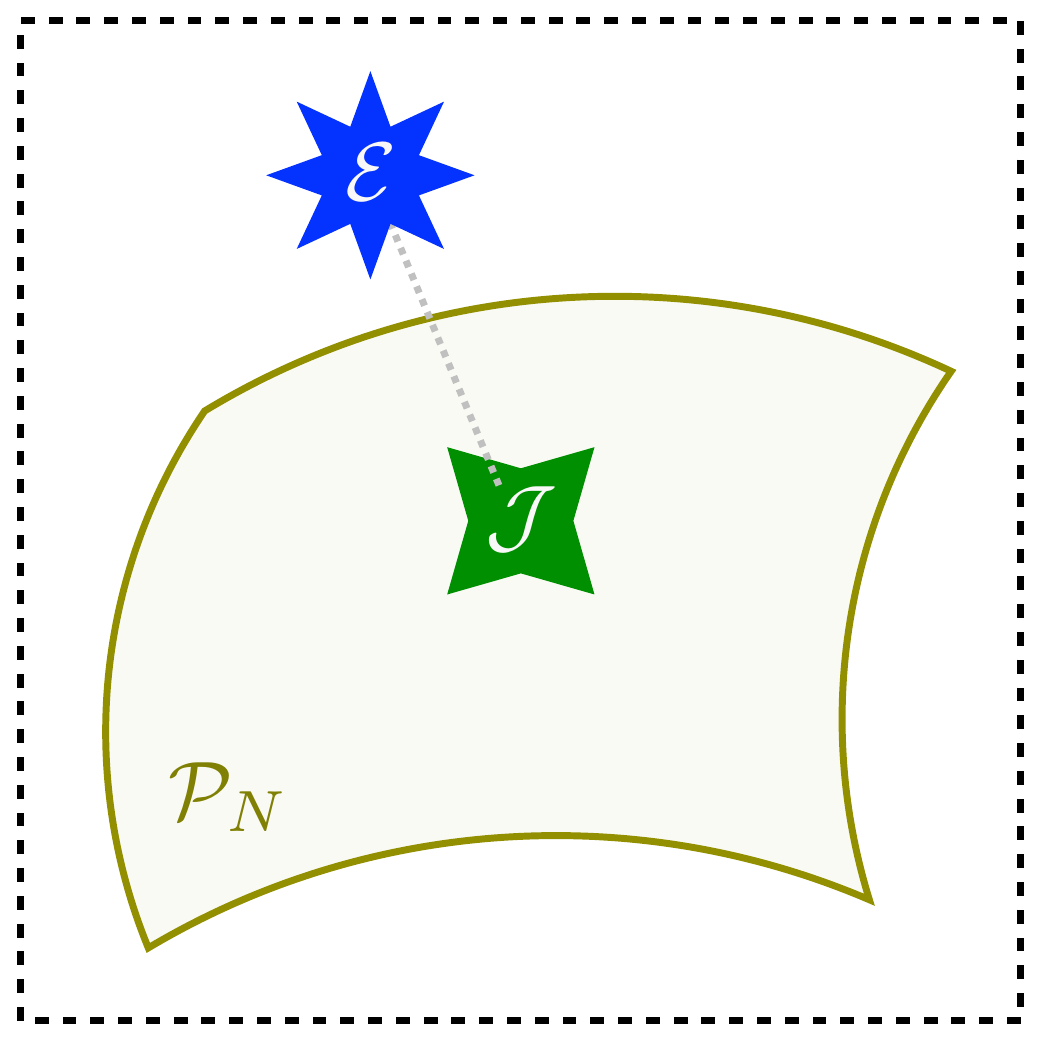}}} \\ 
 \ref{sec:jets} &  {\bf Jets} & $\displaystyle\mathcal J(\E)=\argmin_{\mathcal J\in\mathcal P_N}\,\EMD(\E,\mathcal J)$ &   \\
 \ref{subsec:xcone} &  Cone Finding~\cite{Stewart:2015waa,Thaler:2015xaa}  & & \\
 \ref{sec:seqrec} & Seq. Rec.~\cite{Catani:1993hr,Ellis:1993tq,Bertolini:2013iqa,Salambroadening} &  & \\
& &  & \\ \hline
& &  & \multirow{5}{*}{\raisebox{-7em}{\includegraphics[scale=0.28]{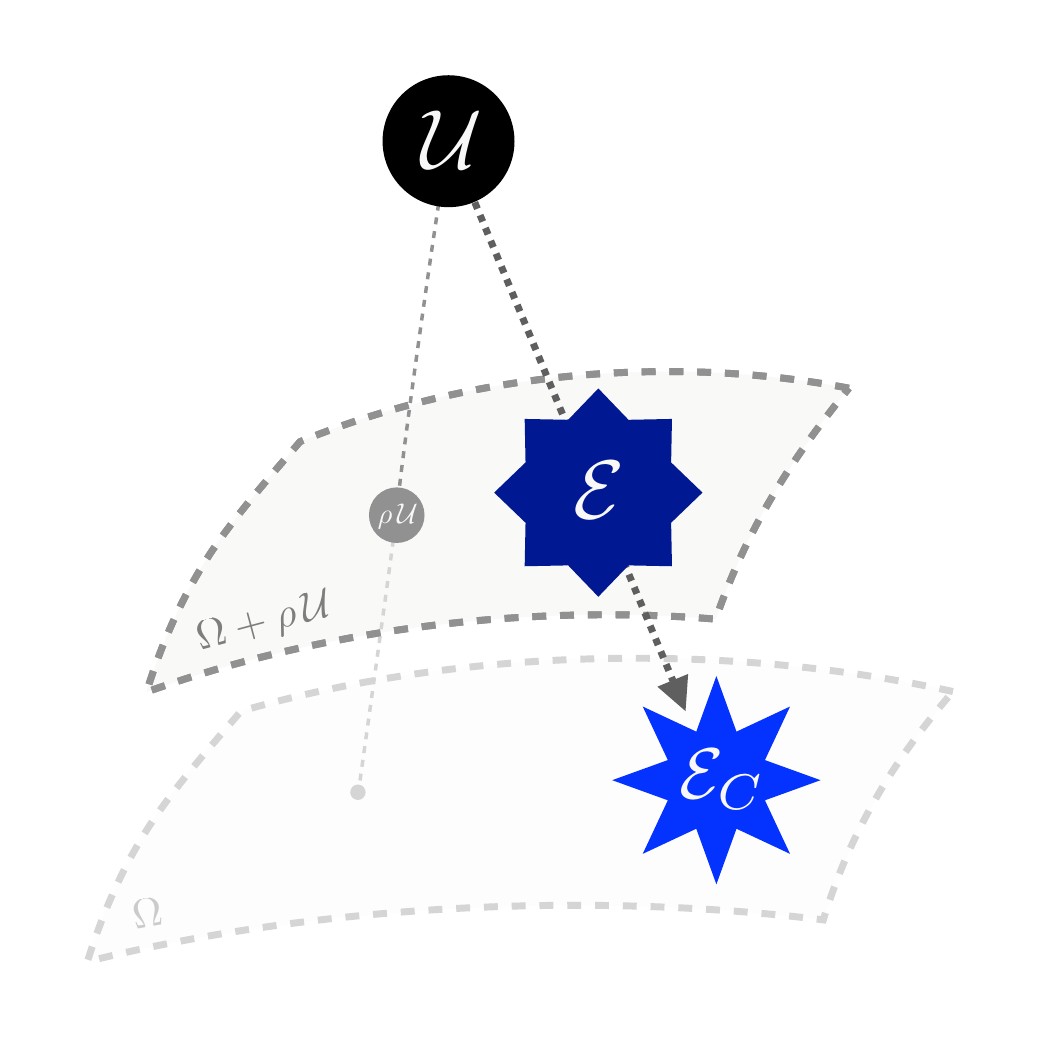}}} \\ 
\ref{sec:pileup} & {\bf Pileup Subtraction} &  $\mathcal E_C(\mathcal E, \rho)  = \displaystyle\argmin_{\mathcal E' \in \Omega}\,\EMD(\mathcal E,\mathcal E' + \rho\,\mathcal U)$ & \\
 &\cite{Cacciari:2007fd,Cacciari:2008gn,Soyez:2012hv,Berta:2014eza,Bertolini:2014bba,Soyez:2018opl,Berta:2019hnj}  &  & \\
& &  & \\
& &  & \\ \hline
& & & \multirow{3}{*}{\raisebox{-7.2em}{\includegraphics[scale=0.28]{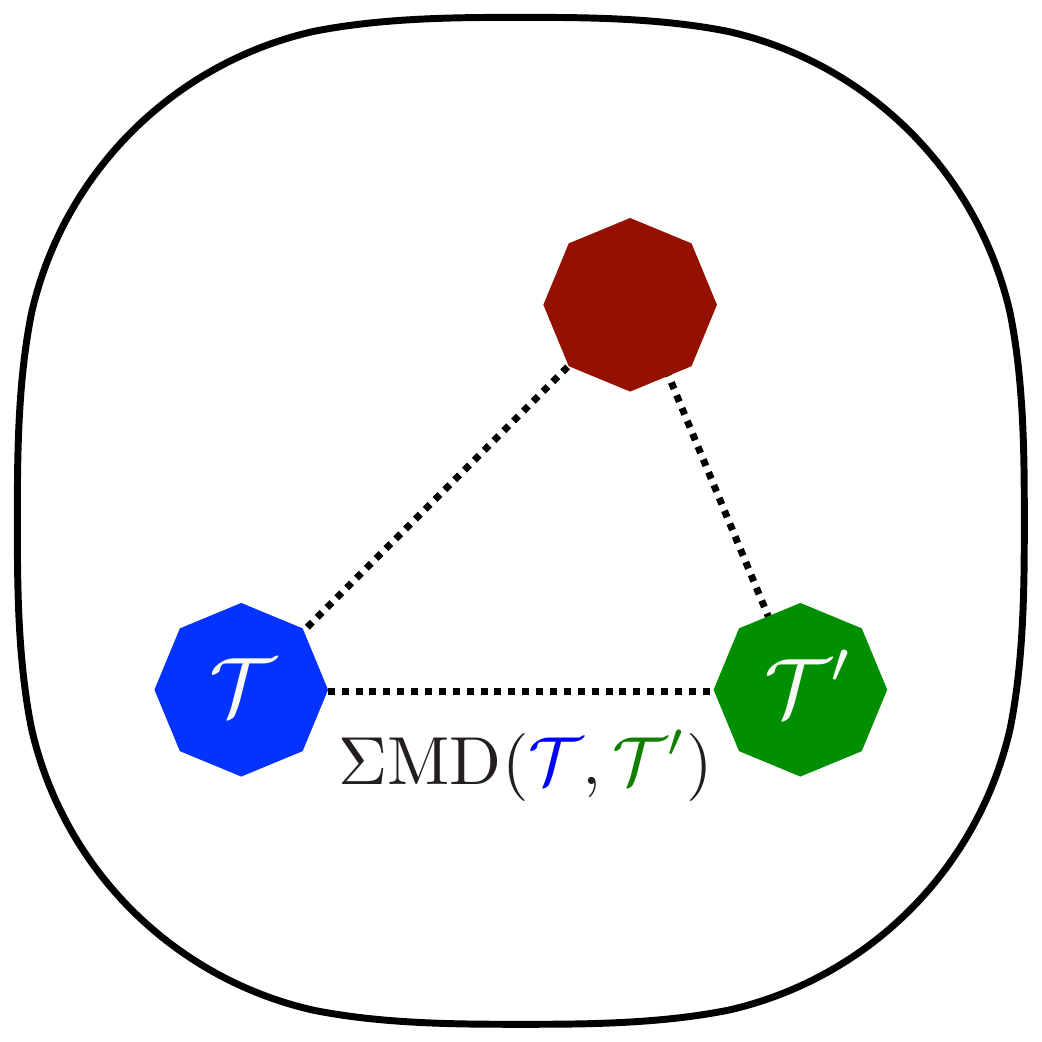}}}  \\
\ref{sec:theory} & {\bf Theory Space} & $\displaystyle\mathcal{T}(\mathcal E) = \sum_{i = 1}^N \sigma_i \, \delta(\mathcal E - \mathcal E_i)$ & \\
& & & \\
& & & \\
\hline\hline
\end{tabular}
\caption{\label{tab:outlinefigs}
Concepts from quantum field theory and collider physics, unified in this paper as geometric and topological constructions in the space of events.
In \Sec{sec:safety}, IRC safety is identified as continuity in this space.
In \Sec{sec:observables}, many classic collider observables are shown to be the shortest distance between the event and a manifold of events.
In \Sec{sec:jets}, popular jet algorithms are derived by projecting the event onto manifolds of $N$-particle events.
In \Sec{sec:pileup}, common pileup mitigation strategies are cast as transporting away uniform radiation.
In \Sec{sec:theory}, a space of theories is developed using a distance between event distributions.
}
\end{table}
\afterpage{\clearpage}

The key concepts unified in this paper are outlined and summarized in \Tab{tab:outlinefigs}.
In \Sec{sec:safety}, we discuss observables as functions defined on event radiation patterns and IRC safety as smoothness in the space of energy flows.
Colloquially, the label ``IRC safe'' indicates that an observable should be well-defined and calculable in perturbation theory~\cite{Kinoshita:1962ur,Lee:1964is} due to its robustness to long-distance effects (e.g.~hadronization in the case of QCD).
This ``perturbatively accessible'' IRC safety is traditionally connected to the observable being ``insensitive'' to the addition of low energy particles or collinear splittings of particles~\cite{Sterman:1977wj,Sterman:1978bi,Sterman:1978bj,Sterman:1979uw,sterman1995handbook,Weinberg:1995mt,Ellis:1991qj,Banfi:2004yd}.
Here, we refine the definition of IRC safety and clarify when discontinuities in an observable spoil its perturbative calculability.
Critical to our formulation is the notion of continuity with respect to the metric topology provided by the EMD:
\begin{definition}\label{def:emdcontinuity}
An observable $\O$ is $\EMD$ continuous at an event $\E$ if, for any $\epsilon>0$, there exists a $\delta>0$ such that for all events $\E'$:
\begin{equation}\label{eq:emdcontinuity}
\text{\emph{EMD}}(\E,\E')<\delta \quad \implies\quad |\O(\E) - \O(\E')| < \epsilon.
\end{equation}
\end{definition}
\noindent We argue that IRC safety is $\EMD$ continuity everywhere except a negligible set of events, where a negligible set is one that contains no EMD balls of non-zero radius.
Using the EMD provides a definition of IRC safety that does not refer to particles directly, which circumvents many pathologies of previous definitions.
We argue that observables that are calculable in fixed-order perturbation are exactly those that satisfy a slightly stronger continuity condition known as H\"older continuity~\cite{Ortega2000,Gilbarg2001}, which restricts the types of divergences that can appear in the distribution of an observable~\cite{Sterman:1979uw,Banfi:2004yd}.
Fascinatingly, this framework naturally accommodates Sudakov-safe observables~\cite{Larkoski:2013paa,Larkoski:2014wba,Larkoski:2015lea} as those that are IRC safe but fail to satisfy EMD H\"older continuity on a non-negligible subset of some $\mathcal P_N$ (where a non-negligible subset of $\mathcal P_N$ is one that has measure in $\mathcal P_N$).
This suggests, in agreement with \Ref{Larkoski:2015lea}, that Sudakov safe observables are indeed perturbatively calculable once properly regulated.

In \Sec{sec:observables}, we highlight that many well-known collider observables can be viewed as the distance of closest approach between an event and a manifold of events.
Many of the observables we consider can be exactly cast as:
\begin{equation}
\label{eq:obsdefemd}
\O(\E) = \min_{\E'\in\mathcal M}\EMD_{\beta,R}(\E,\E'),
\end{equation}
for particular choices of the manifold $\mathcal M$ and parameters $\beta$ and $R$.
Observables that have the form of \Eq{eq:obsdefemd} include thrust~\cite{Brandt:1964sa,Farhi:1977sg}, spherocity~\cite{Georgi:1977sf}, (recoil-free) broadening~\cite{Larkoski:2014uqa}, and $N$-jettiness~\cite{Stewart:2010tn}.
Particularly interesting is the event isotropy, recently proposed in \Ref{isotropytemp}, which was inspired by EMD geometry and is directly based on optimal transport.
This geometric framework also includes jet substructure observables such as jet angularities~\cite{Ellis:2010rwa} and  $N$-subjettiness~\cite{Thaler:2010tr,Thaler:2011gf}.

In \Sec{sec:jets}, we demonstrate how jet finding can be phrased in our geometric language.
Intuitively, a jet algorithm ``approximates'' an $M$-particle event with $N<M$ objects called jets.
To phrase this geometrically, we are interested in the point of closest approach in $\mathcal P_N$ to our event, allowing us to define jets as:
\begin{equation}
\label{eq:jetdefemd}
\mathcal J(\E)=\argmin_{\mathcal J\in\mathcal P_N}\,\EMD_{\beta,R}(\E,\mathcal J),
\end{equation}
where $\mathcal J$ is the collection of $N$ jets corresponding to the event $\E$.
Many common jet finding algorithms can be derived in full detail from \Eq{eq:jetdefemd}.
For instance, we show that jets defined by \Eq{eq:jetdefemd} are precisely those found by XCone~\cite{Stewart:2015waa,Thaler:2015xaa}, where $\beta$ is the angular weighting exponent and $R$ is the jet radius.
Also, several popular sequential clustering algorithms and recombination schemes, such as $k_T$ clustering~\cite{Catani:1993hr,Ellis:1993tq} with winner-take-all recombination~\cite{Bertolini:2013iqa,Larkoski:2014uqa,Salambroadening}, can be exactly obtained by iterating \Eq{eq:jetdefemd} with $N = M-1$ for various $\beta$.
It is satisfying that a rich diversity of jet algorithms can be concisely encoded using event geometry, and we find that several new schemes not previously appearing in the literature naturally emerge.

In \Sec{sec:pileup}, we connect several pileup mitigation strategies to optimal transport through the EMD.
There is a long-established relationship between pileup subtraction and geometric concepts~\cite{Cacciari:2007fd,Cacciari:2008gn,Soyez:2012hv,Berta:2014eza,Bertolini:2014bba,Soyez:2018opl,Berta:2019hnj}.
Since pileup is reasonably modeled as uniform contamination in rapidity and azimuth~\cite{Soyez:2018opl}, we phrase pileup subtraction as removing a uniform distribution of radiation from the event using optimal transport.
Intuitively, pileup mitigation finds the event that, when combined with an amount $\rho$ of uniform radiation $\mathcal U$, is closest to the given event:
\begin{equation}\label{eq:pileupemdintro}
\mathcal E_C(\mathcal E, \rho) = \argmin_{\mathcal E' \in \Omega}\,\EMD_\beta(\mathcal E,\mathcal E' + \rho\,\mathcal U),
\end{equation}
yielding the pileup-corrected event $\mathcal E_C$.
Here, $\Omega$ refers to the space of all possible energy flows and $\EMD_\beta$ compares events of equal energy, as described at the beginning of \Sec{sec:observables}.
We demonstrate that Voronoi area subtraction~\cite{Cacciari:2007fd,Cacciari:2008gn} and constituent subtraction~\cite{Berta:2014eza} can be phrased exactly as \Eq{eq:pileupemdintro} in the small-pileup limit.
Generalizing this to the large-pileup limit, we develop two new pileup subtraction schemes, Apollonius subtraction and iterated Voronoi subtraction, and discuss their prospects and potential advantages.

In \Sec{sec:theory}, we introduce a distance between theories: the cross section mover's distance (stylized as $\Sigma$MD, using the typical greek letter for cross section).
Here, a ``theory'' $\mathcal T$ is taken to be a distribution over (or collection of) events $\{\E_i\}$ weighted by cross sections $\{\sigma_i\}$:
\begin{equation}
\label{eq:T_defintro}
\mathcal{T}(\mathcal E) = \sum_{i = 1}^N \sigma_i \, \delta(\mathcal E - \mathcal E_i).
\end{equation}
The $\Sigma$MD is formulated as an optimal transport problem with EMD as the ground metric and cross sections as the weights.
The similarity of the constructions of EMD and $\Sigma$MD are highlighted in \Tab{tab:emdsmdcomp}.
Interestingly, we connect $\Sigma$MD to a recently proposed technique for probing jet modifications due to the quark-gluon plasma by comparing similar sets of events between proton-proton and heavy-ion collisions~\cite{Brewer:2018dfs}.
We also demonstrate that representative events can be identified by clustering using the $\Sigma$MD, analogously to how particles are clustered into jets.
The $\Sigma$MD provides the foundation for a rigorous formulation of ``theory space'', quantifying how different two theories are based on all of their physically observable quantities simultaneously.

\begin{table}[t]
\centering
\begin{tabular}{rcc}
\hline\hline
& Energy Mover's Distance & Cross Section Mover's Distance  \\ \hline \hline
Symbol & \text{EMD} & \text{$\Sigma$MD} \\
Description & Distance between events & Distance between theories \\ \hline
Weight & Particle energies $E_i$ & Event cross sections $\sigma_i$ \\
Ground Metric & Particle distances $\theta_{ij}$ & Event distances $\text{EMD}(\mathcal E_i, \mathcal E_j)$ \\
\hline\hline
\end{tabular}
\caption{\label{tab:emdsmdcomp}
Comparing the constructions of EMD and $\Sigma$MD as optimal transport problems.
Events are treated as energy-weighted angular distributions, whereas theories are treated as cross section-weighted event distributions.
This connection allows us to bootstrap the EMD as a ground metric for the $\Sigma$MD to develop a rigorous notion of theory space.
}
\end{table}

Our conclusions are presented in \Sec{sec:conc}, where we also highlight the interesting and unique interplay between machine learning and the natural sciences in this story.

\section{Infrared and collinear safety: Smoothness in the space of events}
\label{sec:safety}

\begin{figure}[t]
\centering
\includegraphics[width=\columnwidth]{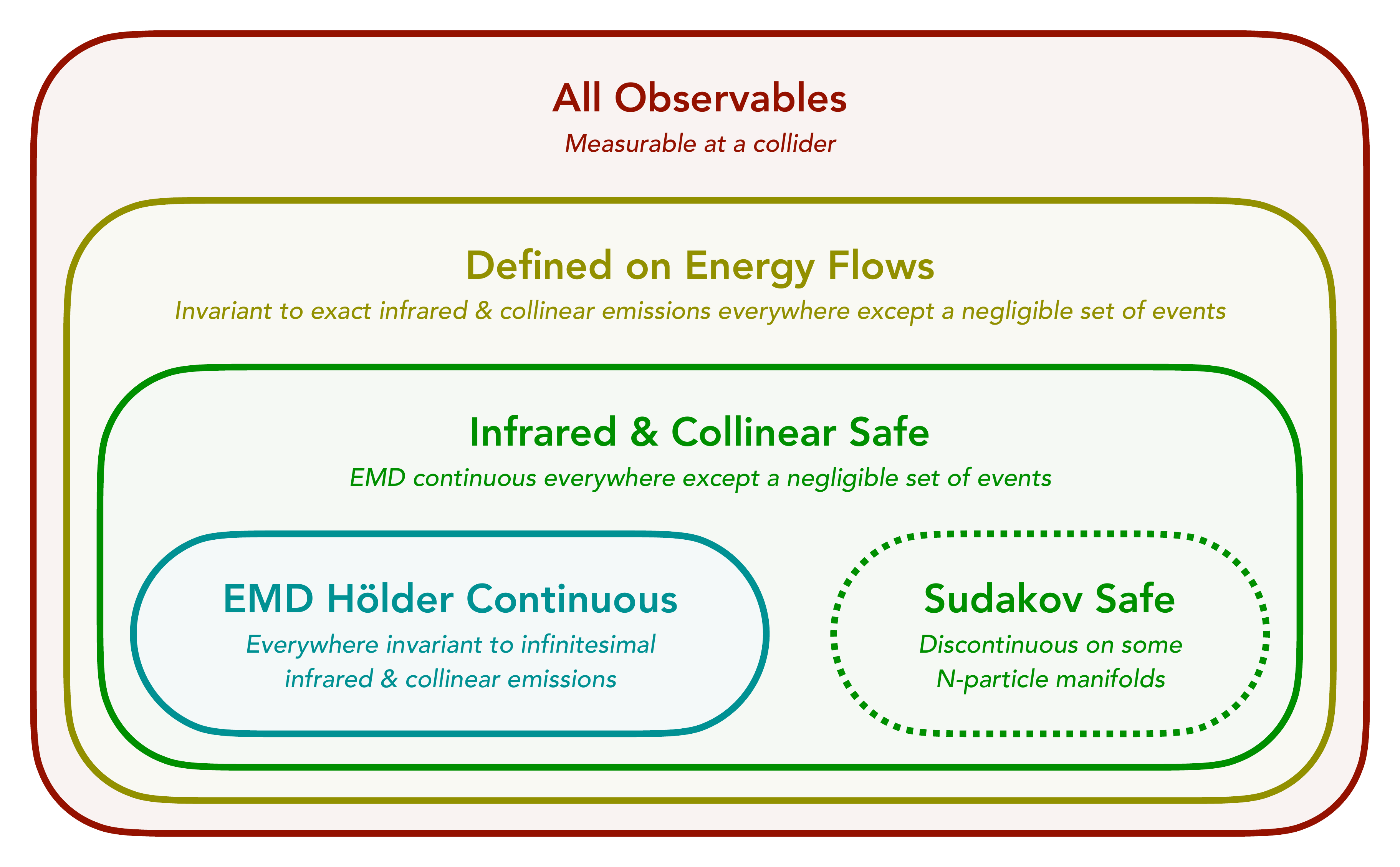}
\caption{\label{fig:obsset}
An illustration of the set of observables partitioned according to various IRC-invariance properties.
Examples of observables in each category are listed in \Tab{tab:exampleobs}.
}
\end{figure}

IRC safety is a central notion in collider physics because it indicates when an observable is robust to long distance effects and hence can be described in perturbation theory~\cite{Kinoshita:1962ur,Lee:1964is} using a combination of fixed-order calculations and resummation.
This insensitivity is frequently connected to the invariance of an observable under certain modifications of the event, namely soft and collinear splittings~\cite{Sterman:1977wj,Sterman:1978bi,Sterman:1978bj,sterman1995handbook,Weinberg:1995mt,Ellis:1991qj,Banfi:2004yd}.

In this section, we review some of the common mathematical statements of this invariance that have appeared in the literature, with the goal of clarifying and categorizing their implications.
We arrive at a simple, unified description of IRC safety and related concepts (including Sudakov safety) as statements about continuity in the space of energy flows.
In \Fig{fig:obsset}, we show the breakdown of observables into broad classes according to our categorization.
A few common examples of each category are given in \Tab{tab:exampleobs}.

\begin{table}[p]
\begin{flushright}

\begin{tabular}{p{0.56\textwidth} p{0.38\textwidth}}
\hline\hline
\cellcolor{table_red_bg} \textbf{\textcolor{table_red}{All Observables}} & \cellcolor{table_red_bg} Comments \\
 \hline
\cellcolor{table_red_bg} Multiplicity $\left(\sum_i 1\right)$ & \cellcolor{table_red_bg}  IR unsafe and C unsafe\\
\cellcolor{table_red_bg} Momentum Dispersion~\cite{CMS:2013kfa} $\left(\sum_i E_i^2 \right)$ & \cellcolor{table_red_bg}  IR safe but C unsafe\\ 
\cellcolor{table_red_bg} Sphericity Tensor~\cite{Bjorken:1969wi} $\left(\sum_i p_i^\mu p_i^\nu \right)$ & \cellcolor{table_red_bg}  IR safe but C unsafe\\
\cellcolor{table_red_bg} Number of Non-Zero Calorimeter Deposits & \cellcolor{table_red_bg}  C safe but IR unsafe \\
  \hline \hline
\end{tabular}

\begin{tabular}{p{0.51\textwidth} p{0.38\textwidth}}
 \multicolumn{2}{l}{\cellcolor{table_yellow_bg} \textbf{\textcolor{table_yellow}{Defined on Energy Flows}}} \\
\hline
\cellcolor{table_yellow_bg} Pseudo-Multiplicity ($\min \{N\, |\, \mathcal T_N = 0\}$)& \cellcolor{table_yellow_bg} Robust to exact IR or C emissions\\
\hline  \hline
\end{tabular}

\begin{tabular}{p{0.46\textwidth} p{0.38\textwidth}}
 \multicolumn{2}{l}{\cellcolor{table_green_bg} \textbf{\textcolor{table_green}{Infrared \& Collinear Safe}}} \\
 \hline
\cellcolor{table_green_bg} Jet Energy $\left(\sum_i E_i \right)$ &\cellcolor{table_green_bg}  Disc.\ at jet boundary\\
\cellcolor{table_green_bg} Heavy Jet Mass~\cite{Clavelli:1981yh} &\cellcolor{table_green_bg}  Disc.\ at hemisphere boundary \\
\cellcolor{table_green_bg} Soft-Dropped Jet Mass~\cite{Dasgupta:2013ihk,Larkoski:2014wba} & \cellcolor{table_green_bg} Disc.\ at grooming threshold\\
\cellcolor{table_green_bg} Calorimeter Activity~\cite{Pumplin:1991kc} ($N_{95}$) & \cellcolor{table_green_bg} Disc.\ at cell boundary\\
  \hline
  \hline
   \multicolumn{2}{l}{\cellcolor{table_green_bg} \textit{\textcolor{table_green}{Sudakov Safe}}} \\
  \hline
\cellcolor{table_green_bg} Groomed Momentum Fraction~\cite{Larkoski:2015lea} ($z_g$) & \cellcolor{table_green_bg} Disc.\ on $1$-particle manifold\\
\cellcolor{table_green_bg} Jet Angularity Ratios~\cite{Larkoski:2013paa} & \cellcolor{table_green_bg} Disc.\ on 1-particle manifold \\
\cellcolor{table_green_bg} $N$-subjettiness Ratios~\cite{Thaler:2010tr,Thaler:2011gf} ($\tau_{N+1} / \tau_{N}$) & \cellcolor{table_green_bg} Disc.\ on $N$-particle manifold\\
\cellcolor{table_green_bg} $V$ parameter~\cite{Banfi:2004yd} (\Eq{eq:Vobs}) & \cellcolor{table_green_bg} H\"{o}lder disc.\ on 3-particle manifold \\
 \hline \hline
\end{tabular}

\begin{tabular}{p{0.41\textwidth} p{0.38\textwidth}}
 \multicolumn{2}{l}{\cellcolor{table_teal_bg} \textbf{\textcolor{table_teal}{EMD H\"{o}lder Continuous Everywhere}}} \\
\hline
\cellcolor{table_teal_bg} Thrust~\cite{Brandt:1964sa,Farhi:1977sg} & \cellcolor{table_teal_bg} \\
\cellcolor{table_teal_bg} Spherocity~\cite{Georgi:1977sf} & \cellcolor{table_teal_bg}  \\ 
\cellcolor{table_teal_bg} Angularities~\cite{Berger:2003iw} & \cellcolor{table_teal_bg}  \\ 
\cellcolor{table_teal_bg} $N$-jettiness~\cite{Stewart:2010tn} $\left(\mathcal T_N\right)$ & \cellcolor{table_teal_bg}  \\
\cellcolor{table_teal_bg} $C$ parameter~\cite{Parisi:1978eg,Donoghue:1979vi,Ellis:1980wv,Catani:1997xc} & Resummation beneficial at $C = \frac34$ \cellcolor{table_teal_bg} \\
\cellcolor{table_teal_bg} Linear Sphericity~\cite{Donoghue:1979vi} $\left(\sum_i E_i n_i^\mu n_i^\nu \right)$ & \cellcolor{table_teal_bg}  \\
\cellcolor{table_teal_bg} Energy Correlators~\cite{Banfi:2004yd,Larkoski:2013eya,Larkoski:2014gra,Moult:2016cvt} & \cellcolor{table_teal_bg}  \\
\cellcolor{table_teal_bg} Energy Flow Polynomials~\cite{Komiske:2017aww,Komiske:2019asc} & \cellcolor{table_teal_bg}  \\
\hline\hline
\end{tabular}
\end{flushright}
\caption{
Examples of well-known collider observables, along with their classification according to \Fig{fig:obsset}.
The observables satisfy the conditions of all bold-faced categories above them in the table.
Note that via our classification, Sudakov safe observables are IRC safe, since the discontinuities appear on $N$-particle manifolds which are negligible sets in the full space.
}
\label{tab:exampleobs}
\end{table}

\subsection{Review of infrared and collinear invariance}
\label{sec:invariance}

The most straightforward statement of IRC invariance is that an observable $\mathcal O$ is unchanged under the addition of an exactly zero energy particle or an exactly collinear splitting~\cite{sterman1995handbook}:
\begin{align}
\label{eq:exactirsafety}\text{Exact Infrared Invariance:}&\quad \mathcal O(p_1^\mu, \ldots,p_M^\mu) =  \mathcal O(0 p_0^\mu, p_1^\mu,  \ldots, p_M^\mu),\\
\label{eq:exactcsafety}\text{Exact Collinear Invariance:}&\quad \mathcal O(p_1^\mu,  \ldots, p_M^\mu) = \mathcal O(\lambda p_1^\mu, (1-\lambda) p_1^\mu,\ldots, p_M^\mu),
\end{align}
for any soft momentum $p_0^\mu$ and collinear splitting fraction $\lambda\in[0,1]$.
These conditions correctly rule out some observables from having a perturbative description, such as the number of particles in an event, which change by a finite amount under any splitting.
Exact IRC invariance, however, is not sufficiently restrictive to guarantee perturbative calculability of an observable.
For instance, the number of calorimeter cells with non-zero energy is safe according to \Eqs{eq:exactirsafety}{eq:exactcsafety}, though it is highly sensitive to arbitrarily low-energy effects~\cite{Pumplin:1991kc}.
Similarly, the pseudo-multiplicity, which we define as the smallest $N$ that yields zero $N$-jettiness (see \Sec{subsec:nsubjettiness} below), is unchanged by exact infrared and collinear emissions,%
\footnote{We thank Andrew Larkoski for discussions related to this point.}
but is highly sensitive to any emissions at finite energy or angle.

Another common statement of IRC invariance refines the concept by invoking the limit as particles become soft or collinear~\cite{Sterman:1978bi,Sterman:1978bj,Weinberg:1995mt,Banfi:2004yd}:
\begin{align}
\label{eq:nearirsafety}\text{Near Infrared Invariance:}&\quad \mathcal O(p_1^\mu, \ldots,p_M^\mu) =  \lim_{\epsilon\to0}\mathcal O(\epsilon p_0^\mu, p_1^\mu,  \ldots, p_M^\mu),\\
\label{eq:nearcsafety}\text{Near Collinear Invariance:}&\quad \mathcal O(p_1^\mu, \ldots, p_M^\mu) = \lim_{p_0^\mu\to p_1^\mu}\mathcal O(\lambda p_{0}^\mu, (1-\lambda) p_1^\mu, \ldots, p_M^\mu).
\end{align}
One issue with this definition is that many reasonable observables that have hard boundaries in phase space are excluded, such as jet kinematics due to sensitivity to particles on a jet boundary.
Hybrid definitions mixing exact and near IRC invariance also appear in the literature but they suffer from the same pathologies.
Another issue is that \Eqs{eq:nearirsafety}{eq:nearcsafety} (and also \Eqs{eq:exactirsafety}{eq:exactcsafety}) do not guarantee insensitivity to multiple soft or collinear splittings.

Several of these issues were previously identified in \Ref{Banfi:2004yd}, which utilized a limit-based statement of IRC invariance, recognized the importance of allowing for multiple soft and collinear emissions, and allowed for exceptions on sets of measure zero.
Despite noting that a rigorous mathematical definition of IRC safety would be desirable, \Ref{Banfi:2004yd} concluded that formulating one without pathologies was challenging and that a satisfactory definition had not yet been obtained.
Here, we explore how the geometric picture provided by the EMD yields a natural and elegant way to phrase IRC safety and to control these various subtleties.
This builds on the notion of ``$C$-continuity'' advocated for in \Refs{Tkachov:1995kk,Tkachov:1999py}, which argue that the perturbative calculability of $C$-continuous observables can be seen by relating the energy flow to the stress-energy tensor of the underlying quantum field theory.

\subsection{Infrared and collinear safety in the space of events}
\label{sec:ircsafety}

The EMD provides a natural language for understanding IRC-safe observables as continuous functions on the space of events.
To make this precise, we first must understand which observables are well-defined functions of the energy flow.

We can show that observables that are defined on \emph{all} energy flows are precisely those which have exact IRC invariance according to \Eqs{eq:exactirsafety}{eq:exactcsafety}.
First, an observable is well defined on the space of energy flows if its value is the same on events that are zero EMD apart. 
The following lemma establishes the remaining connection to exact IRC invariance.
\begin{lemma}
Two events are zero EMD apart if and only if they differ by zero energy emissions or exactly collinear splittings.
\end{lemma}
\begin{proof}
Adding a zero energy particle or a collinear splitting to an event manifestly does zero energy moving, proving the forward direction.
To prove the reverse direction, suppose that two events are zero EMD apart and take their energy flows to be:
\begin{equation}
\mathcal E(\hat n) = \sum_{i=1}^M E_i\,\delta(\hat n - \hat n_i),\quad\quad \mathcal E'(\hat n) = \sum_{j=1}^{M'} E'_j \,\delta(\hat n - \hat n'_j).
\end{equation}
Since the EMD is a proper metric between energy flows, the identity of indiscernibles says that $\EMD(\E(\hat n),\E'(\hat n))=0$ implies $\E(\hat n)=\E'(\hat n)$.
For any direction $\hat n$ with at least one particle, either the sums of energies in that direction are equal between the two events or the particle has zero energy.
In the first case, the events differ by exactly collinear splittings in that direction, and in the second case they differ by zero energy particles.
\end{proof}
By this lemma we see that exact IRC invariance ensures that we can write $\mathcal O(\mathcal E)$ rather than $\mathcal O(p_1^\mu, \cdots, p_M^\mu)$ for an observable.
As discussed in \Sec{sec:invariance}, exact IRC invariance is insufficient to guarantee IRC safety and we must formulate a stronger condition phrased in the geometric language of the space of events.

We propose that IRC safety is achieved by requiring an observable to be EMD continuous, in the sense of Definition~\ref{def:emdcontinuity}, except possibly on a negligible set of events.
We define a negligible set to be one that contains no EMD ball.
The (open) \EMD ball $B_r(\E)$ around an event $\mathcal E$ is defined as all events within an EMD of $r>0$:
\begin{equation}
B_r(\E)=\left\{\E'\in\Omega\,\Big|\,\EMD(\E,\E')<r\right\},
\end{equation}
where $\Omega$ is the space of all energy flows.
Implicit in the above requirement is that an observable must be well defined on energy flows.
Concretely, we state IRC safety as the following:
\begin{framed}
\begin{ircsafety}
An observable is IRC safe if it is EMD continuous for all energy flows, except potentially on a negligible set of events.
\end{ircsafety}
\end{framed}
\noindent This new formulation of IRC safety has many aspects of existing ideas of safety discussed in \Sec{sec:ircsafety} wrapped into a concise and rigorous statement.
It makes mathematically precise the intuitive notion that small perturbations in the energy flow of the event give rise to small perturbations in the observable.
This notion of EMD continuity for IRC safe observables is illustrated in \Fig{fig:space_ircsafe}.
The exception for negligible sets allows observables to be discontinuous in a way that affords them the opportunity to depend sharply on phase space but does not spoil their calculability.
Calculability is a statement about integrability, and removing a negligible set of points from an integral cannot change its value.

\begin{figure}[t]
\centering
\includegraphics[scale=0.75]{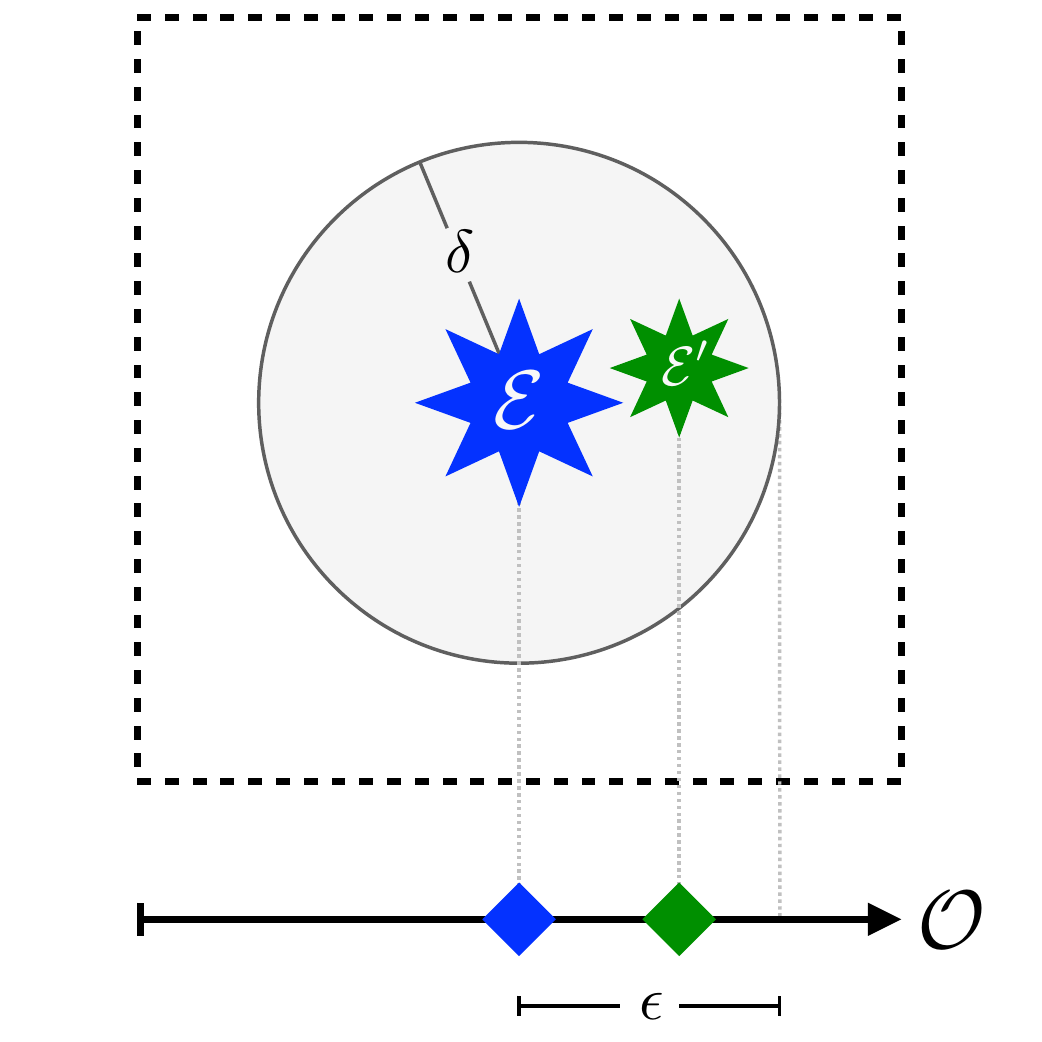}
\caption{\label{fig:space_ircsafe} An illustration of IRC safety of an observable as continuity in the space of events.
As formulated in \Eq{eq:emdcontinuity}, small perturbations to the event, as measured by EMD, yield small changes in the observable value.}
\end{figure}

To get some familiarity with this definition, consider additive IRC-safe observables, which are ubiquitous structures~\cite{Komiske:2019asc} that take the form $\mathcal O(\mathcal E) = \sum_{i=1}^M E_i f(\hat n_i)$ for an angular function $f$.
One can prove that they are Lipschitz continuous in the space of events assuming $f$ is Lipschitz continuous~\cite{Komiske:2019fks}, and therefore they naturally satisfy continuity according to the EMD.
As a generalization of additive observables, energy flow networks~\cite{Komiske:2018cqr} are a machine learning architecture that can approximate any IRC-safe observable through an additive IRC-safe latent space.
As long as the activation functions are continuous almost everywhere, then the final energy flow network output will be IRC safe.

There are also observables that fail the criteria of \Eqs{eq:nearirsafety}{eq:nearcsafety} for small sets of events but are safe according to our definition and are indeed calculable.
The energy of a jet is a simple example where emissions on the jet boundary result in discontinuous behavior of the observable, but this discontinuity is integrable in fixed-order perturbation theory.
A more complicated example is the invariant mass after soft drop grooming~\cite{Larkoski:2014wba,Dasgupta:2013ihk}: for events on the threshold of having an emission dropped, tiny perturbations can give rise to discontinuously large changes in the observable.
This issue, however, only occurs on a negligible set, satisfying our definition of safety and avoiding serious analytic pathologies~\cite{Frye:2016okc,Frye:2016aiz,Marzani:2017mva,Marzani:2017kqd}.
Piecewise continuity does, however, complicate analyzing the nonperturbative corrections~\cite{Hoang:2019ceu} and detector response~\cite{ATL-PHYS-PUB-2019-027,Aad:2019vyi} of soft-dropped jet mass.

Our definition also includes observables that would sometimes not be called IRC safe since they do not have a well defined Taylor expansion in the small parameter of the theory (e.g.\ $\alpha_s$ for QCD).
These observables are nevertheless perturbatively calculable, though methods beyond fixed-order perturbation theory may be required.
The next subsections are devoted to exploring which IRC-safe observables are calculable in fixed-order perturbation theory and which require additional techniques.

\subsection{Calculability in fixed-order perturbation theory}
\label{sec:fopt}

IRC safety has long been connected with the notion of calculability order-by-order in perturbative quantum field theory.
However, IRC safety according to our Definition~\ref{def:emdcontinuity} includes observables that are not calculable in fixed-order perturbation theory, which we explore further in the next subsection.
Here, building off the work in \Refs{Sterman:1979uw,Banfi:2004yd}, we formulate the stronger notion of EMD H\"older continuity~\cite{Ortega2000,Gilbarg2001} and argue that it is the appropriate condition to guarantee order-by-order perturbative control:
\begin{definition}\label{def:emdholdercontinuity}
An observable $\mathcal O$ is EMD H\"older continuous with exponent $\alpha\in(0,1]$ at an event $\E$ if there exists $K>0$ such that for all $\E'$ in some neighborhood of $\E$:
\begin{equation}
\label{eq:Holder}
|\mathcal O(\E)-\mathcal O(\E')|\le K\,\text{\emph{EMD}}(\E,\E')^\alpha.
\end{equation}
\end{definition}
\noindent Note that the case of $\alpha=1$ corresponds to Lipschitz continuity at $\E$, and in general we have containment such that H\"older continuity with exponent $\alpha$ implies H\"older continuity with exponent $\beta$ if $\beta\le\alpha$.
EMD H\"older continuity effectively specifies that the $\delta$ in Definition~\ref{def:emdcontinuity} is no smaller than $\epsilon$ to some power (times a constant) for all points in a neighborhood of $\E$, and thus it is a stronger requirement than plain EMD continuity.

To connect to fixed-order perturbation theory, we state the following conjecture:
\begin{framed}
\begin{conjecture}\label{conj:fopt}
An observable is calculable order-by-order in perturbation theory if it is EMD H\"older continuous on all but a negligible set of events in each $N$-particle manifold. 
\end{conjecture}
\end{framed}
\noindent This relation phrases the ideas of \Ref{Sterman:1979uw} and ``Version 2'' of the IRC safety definition of \Ref{Banfi:2004yd} in our geometric language via the EMD.
While these criteria were originally formulated for the calculability of moments of an observable, they appear to also extend to the calculability of distributions of observables~\cite{Sterman:2006uk}.

It is possible to demonstrate a precise equivalence between our Conjecture~\ref{conj:fopt} and the following criteria of \Ref{Sterman:1979uw} regarding when the average value of an observable $\O$ is calculable in fixed-order perturbation theory:
\begin{equation}
\label{eq:stermanenergycond}
\lim_{|\vec p_i|\to0}\frac{\O(\vec p_1,\ldots,\vec p_i,\ldots,\vec p_M)-\O(\vec p_1,\ldots,\vec p_{i-1},\vec p_{i+1},\ldots,\vec p_M)}{|\vec p_i|^a}=0,
\end{equation}
\begin{equation}
\label{eq:stermanthetacond}
\lim_{\theta_{ij}\to0}\frac{\O(\vec p_1,\ldots,\vec p_i,\ldots,\vec p_j,\ldots, \vec p_M)-\O(\vec p_1,\ldots,\vec p_i+\vec p_j,\ldots,\vec p_{j-1},\vec p_{j+1},\ldots, \vec p_M)}{\theta_{ij}^b}=0,
\end{equation}
where the powers $a$ and $b$ are positive and the choices of $i$ and $j$ are arbitrary.
Here, \Eq{eq:stermanenergycond} is a statement of H\"older continuity in the energy of particle $i$, which implies ordinary soft safety.
Similarly, \Eq{eq:stermanthetacond} is a statement of H\"older continuity in the angular distance between particles $i$ and $j$, which implies ordinary collinear safety.
In these soft and collinear limits, $\EMD(\mathcal E, \mathcal E') \propto E_i$ and $\EMD(\mathcal E, \mathcal E') \propto \theta_{ij}$ respectively, and so \Eqs{eq:stermanenergycond}{eq:stermanthetacond} can be phrased compactly as:
\begin{equation}
\lim_{\mathcal E' \to \mathcal E} \frac{\mathcal O(\mathcal E) - \mathcal O(\mathcal E')}{\EMD(\mathcal E, \mathcal E')^c} = 0.
\end{equation}
for some positive exponent $c$.
This is equivalent to the H{\"{o}}lder continuity of the observable $\mathcal O$ at $\E$ with some exponent $\alpha \ge c$, connecting the formulation of \Ref{Sterman:1979uw} to our conjecture.

Our Conjecture~\ref{conj:fopt} also nicely connects to ``Version 2'' of the IRC safety definition in \Ref{Banfi:2004yd}, which we restate here with a suggestive relabeling of the original notation.
The criteria for fixed-order calculability of an observable in \Ref{Banfi:2004yd} are as follows:
\begin{quote}
\Ref{Banfi:2004yd}: Given almost any fixed set of particles and any value $n$, then for any $\epsilon>0$, however small, there should exist a $\delta>0$ such that producing $n$ extra soft or collinear emissions, each emission being at a distance of no more than $\delta$ from the nearest particle, then the value of the observable does not change by more than $\epsilon$. Furthermore, there should exist a positive power $c$ such that for small $\epsilon$, $\delta^c$ can always be taken greater than $\epsilon$.
\end{quote}
By equipping the space of events with these topological and geometric structures via EMD, our language provides a natural language to sharply mathematically formulate this discussion.
The first sentence can be encoded as EMD continuity of the observable on all but a negligible set of events.
The power relation between the $\epsilon$ and $\delta$ parameters is precisely captured by EMD H\"{o}lder continuity with some exponent $\alpha > c$, connecting to our Conjecture~\ref{conj:fopt}.

A variety of observables are considered in \Ref{Banfi:2004yd} at the boundary of perturbative calculability, which helpfully illustrate the various requirements in their definition.\footnote{We thank Gavin Salam for discussions related to this point.}
An observable that is useful to consider is:
\begin{equation}\label{eq:Vobs}
V(\mathcal E) = \mathcal T_2(\mathcal E) \left(1 + \frac{1}{\ln E(\mathcal E)/{\mathcal T_3(\mathcal E)}}\right),
\end{equation}
where $\mathcal T_N$ are $N$-jettiness observables~\cite{Stewart:2010tn} discussed further in \Sec{sec:njettiness}, and $E$ is the total energy of the event.
We will refer to this observable as the ``$V$ parameter''.
The double logarithmic structure of $\mathcal T_3$ spoils the integrability of $V$ at fixed order due to its behavior as $\mathcal T_3$ goes to zero~\cite{Banfi:2004yd}, which occurs on the three-particle manifold $\mathcal P_3$.
Nonetheless, this observable can be calculated using techniques beyond fixed-order perturbation theory, such as the Sudakov safety approach discussed in the next section.

The relation between our formalism and fixed-order perturbative calculability is phrased as a conjecture since additional subtleties or nuances about this type of calculability may emerge with future research.
Nonetheless, it is very satisfying that our geometric language provides an efficient encapsulation and unification of the existing formulations of \Refs{Sterman:1979uw,Banfi:2004yd}.
In future work, it would be interesting to find a geometric phrasing of recursive IRC safety~\cite{Banfi:2004yd}, which is a more restrictive condition than EMD H\"older continuity and relevant for understanding factorization and resummation.
It would also be interesting to find a geometric phrasing of unsafe observables that can be nevertheless be computed with the help of non-perturbative fragmentation functions (see \Ref{Elder:2017bkd} for a broad class of such observables).
We hope that further refinements and developments will benefit from and be enabled by the rigorous geometric and topological constructions we have introduced for the space of events via the EMD.

\subsection{A refined understanding of Sudakov safety}

Sudakov-safe observables~\cite{Larkoski:2013paa,Larkoski:2014wba,Larkoski:2015lea} are an interesting class of observables that are not typically considered IRC safe because divergences may appear order by order in perturbation theory; this issue was originally pointed out in \Ref{Soyez:2012hv}.
Nevertheless, the distribution for a Sudakov-safe observable $\mathcal{O}_s$ can be computed perturbatively by calculating its conditional distribution with an IRC-safe companion observable $\mathcal{O}_c$, resumming the $\mathcal{O}_c$ distribution, and then marginalizing over $\mathcal{O}_c$ to obtain a finite answer~\cite{Larkoski:2015lea}:
\begin{equation}
\label{eq:sudakov_safe_strategy}
p(\mathcal{O}_s) = \int \text{d}\mathcal{O}_c \, p(\mathcal{O}_s | \mathcal{O}_c) \, p(\mathcal{O}_c).
\end{equation}
The conditional probability $p(\mathcal{O}_s | \mathcal{O}_c)$ can either be computed in fixed-order perturbation theory or it can be further resummed to obtain a more accurate prediction for $p(\mathcal{O}_s)$.

Here, we interpret Sudakov-safe observables as observables that \emph{are} IRC safe according to our definition but may be EMD (H\"{o}lder) discontinuous on sets with non-zero measure when restricted to some idealized massless $N$-particle manifold $\mathcal P_N$, defined in \Eq{eq:npmanifold}.
The relevant manifolds are the $N$-particle manifolds since these contain the infrared singular regions of massless gauge theories, namely configurations that differ by soft and collinear splittings.
The IRC safety of an observable according to our definition guarantees that any potentially problematic energy flows are infinitesimally close to energy flows for which the observable is well defined.
The strategy in \Eq{eq:sudakov_safe_strategy} also enables the computation of observables such as the $V$ parameter in \Eq{eq:Vobs}, which are EMD continuous everywhere but exhibit H\"{o}lder discontinuities on sets with non-zero measure in $\mathcal P_N$ and are therefore incalculable with fixed-order perturbation theory alone.

It is instructive to make a connection to practical methods of computing Sudakov-safe observables.
In a quantum field theory of massless particles, the cross section to produce events with exactly $N$ particles is zero (i.e.~the naive $S$-matrix is zero), and such theories ultimately yield smooth predictions in the space of events.
Hence, divergences that appear in the calculation of such an observable in a fixed-order expansion can be regulated by a joint, all-orders calculation of the observable and the distance from the problematic manifold $\mathcal P_N$.
This is precisely the strategy represented by \Eq{eq:sudakov_safe_strategy}, though \Ref{Larkoski:2015lea} did not provide a generic method to identify the companion observable $\mathcal{O}_c$.
In \Sec{sec:observables}, we will establish that the distance from an event to the manifold $\mathcal P_N$ is precisely $N$-(sub)jettiness~\cite{Stewart:2010tn,Thaler:2010tr,Thaler:2011gf}, suggesting that they are universal companion observables for the calculation of Sudakov-safe observables, in a similar spirit to \Refs{Alioli:2012fc,Alioli:2013hqa,Alioli:2015toa}.

It is worth mentioning that, even if an observable is EMD H\"{o}lder continuous everywhere, resummation along the lines of \Eq{eq:sudakov_safe_strategy} may still be beneficial for making reliable predictions.
The $C$-parameter~\cite{Parisi:1978eg,Donoghue:1979vi,Ellis:1980wv} is an example of an EMD H\"{o}lder continuous observable, yet its fixed-order perturbative distribution exhibits discontinuous behavior at $C = \frac{3}{4}$~\cite{Catani:1997xc}.
This perturbative discontinuity can be smoothed through soft-gluon resummation, and such techniques are relevant for other observables that exhibit Sudakov shoulder behavior~\cite{Larkoski:2015uaa}.
This is different, however, from Sudakov-safe observables, where the observable itself (and not just its distribution) is ill-defined on some $\mathcal P_N$.

To summarize, our definition of IRC safety does includes Sudakov-safe observables, but we argue that this is appropriate since such observables are indeed perturbatively accessible via regulation with $N$-(sub)jettiness.
This motivates the following conjecture:
\begin{framed}
\begin{conjecture}\label{conj:safe_is_calculable}
An observable is perturbatively calculable, using a combination of fixed-order and resummation techniques, if it is IRC safe according to the definition in \Sec{sec:ircsafety}.
\end{conjecture}
\end{framed}
\noindent Proving this conjecture, or finding a counterexample, would shed considerable light on the structure of perturbative quantum field theory.
Of course, even if an observable is perturbatively calculable, it may suffer from large non-perturbative or detector corrections, and it may be helpful to use the $K$ and $\alpha$ parameters in \Eq{eq:Holder} to assess the sensitivity of observables to long-distance effects.

\section{Observables: Distances between events and manifolds}
\label{sec:observables}

In this section, we show that a number of event-level and jet substructure observables can be identified as geometric quantities in the space of events.
Broadly speaking, the observables we consider take the general form of a distance between an event and a manifold, as in \Eq{eq:obsdefemd}.
The illustration in \Fig{fig:space_obs} shows an observable as a distance between geometric objects in the space of events.
While not all IRC-safe observables can be written in this way, a remarkably large family of classic observables take precisely this geometric form.
We will work with unnormalized observables here, but normalized versions can be obtained by dividing by the total energy (or transverse momentum in the hadronic case).

\begin{figure}[t]
\centering
\includegraphics[scale=0.6]{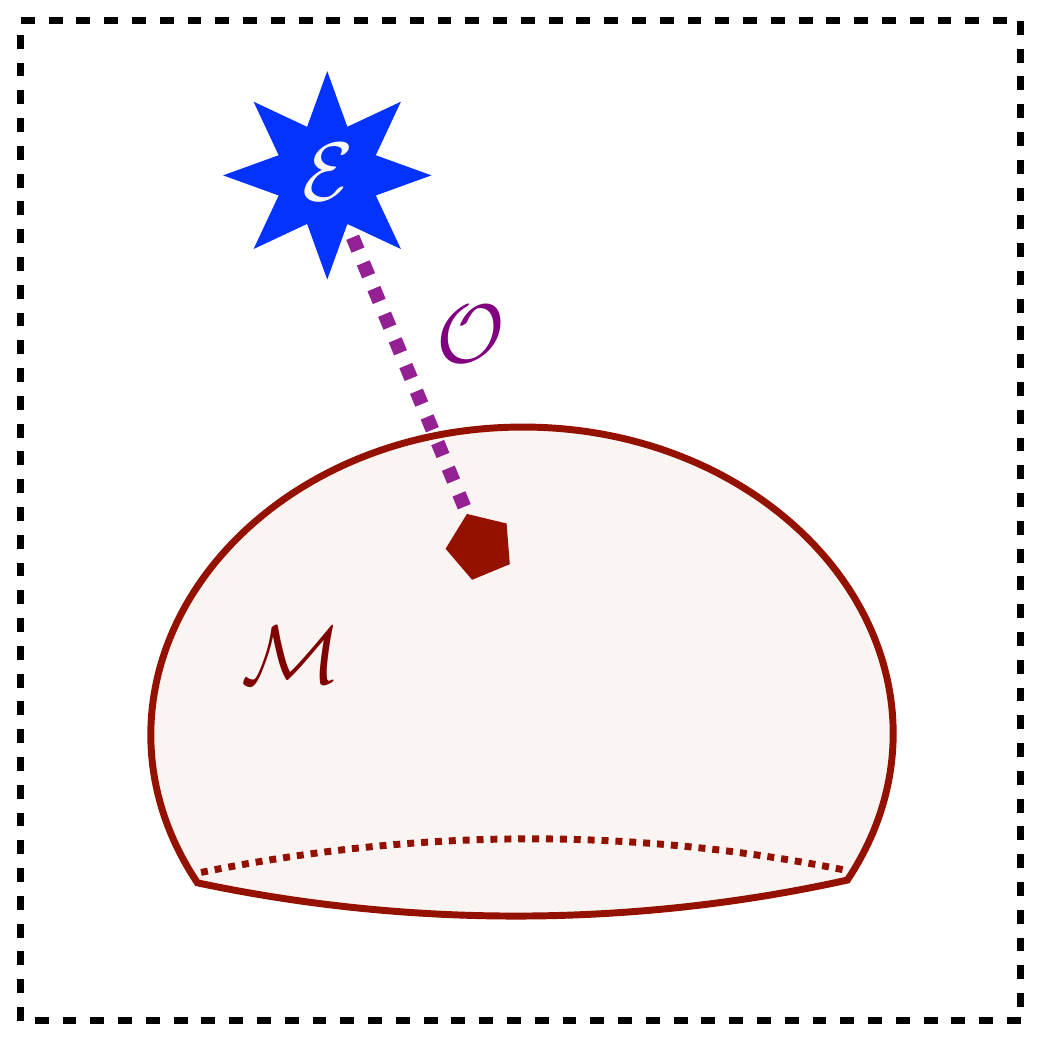}
\caption{\label{fig:space_obs} An illustration of an observable $\mathcal O$ as the distance of closest approach, as measured by the EMD, between the event $\mathcal E$ and a manifold $\mathcal M$ of events.
Many classic collider observables fit into this precise form, stated in \Eq{eq:obsdefemd}, with particular choices of manifold.
}
\end{figure}

We begin by discussing thrust and spherocity, where the manifold is the set of all back-to-back two-particle events.
To understand (recoil-free) broadening, we expand the manifold to all two-particle events, beyond just back-to-back configurations.
Then, to connect to $N$-jettiness, we utilize the idealized $N$-particle manifold defined in \Eq{eq:npmanifold}.
Our geometric language gives clear and intuitive explanations of what physics these observables probe and why they take the forms that they do.
While these EMD formulations do not necessarily lead to practical computational improvements, we do highlight ways to speed up the numerical evaluation of event isotropy using techniques from the optimal transport literature.
Finally, we identify jet angularities and $N$-subjettiness as jet substructure observables obeying similar principles at the level of jets.

\begin{table}[p]
\centering
\begin{tabular}{rccl}
\hline
\hline
  & \multicolumn{3}{l}{$\mathcal O(\mathcal E)=\displaystyle \min_{\mathcal E'\in\mathcal M} \text{EMD}_\beta(\mathcal E, \mathcal E')$} \\
 Name &  & $\beta$ & Manifold $\mathcal M$ \\ \hline\hline
Thrust & $t(\mathcal E)$ & 2& $\mathcal P^{\rm BB}_2$:  2-particle events, back to back \\
Spherocity & $\sqrt{s(\mathcal E)}$ & 1 &  $\mathcal P^{\rm BB}_2$:  2-particle events, back to back \\ 
Broadening & $b(\mathcal E)$ & 1 & $\mathcal P_2$:  2-particle events \\ 
$N$-jettiness & $\mathcal T_N^{(\beta)} (\mathcal E)$& $\beta$ & $\mathcal P_N$:  $N$-particle events \\ 
Isotropy & $\mathcal I^{(\beta)}(\mathcal E)$ & $\beta$ & $\mathcal M_{\mathcal{U}}$: Uniform events\\
\hline
Jet Angularities & $\lambda_\beta(\mathcal J)$ & $\beta$ & $\mathcal P_1$:   1-particle jets\\ 
$N$-subjettiness & $\tau_N^{(\beta)}(\mathcal J)$ & $\beta$ & $\mathcal P_N$:   $N$-particle jets \\ 
\hline
\hline
\end{tabular}
\caption{Observables as the EMD between the event $\mathcal E$ and a manifold $\mathcal{M}$, using the EMD definition in \Eq{eq:emd_noR}.
Several of these observables are illustrated in \Fig{fig:space_manyobs}.
Here, we consider only the ``recoil-free'' versions of these observables.}
\label{tab:obs}
\end{table}

\begin{figure}[p]
\centering
\subfloat[]{\includegraphics[scale=0.685]{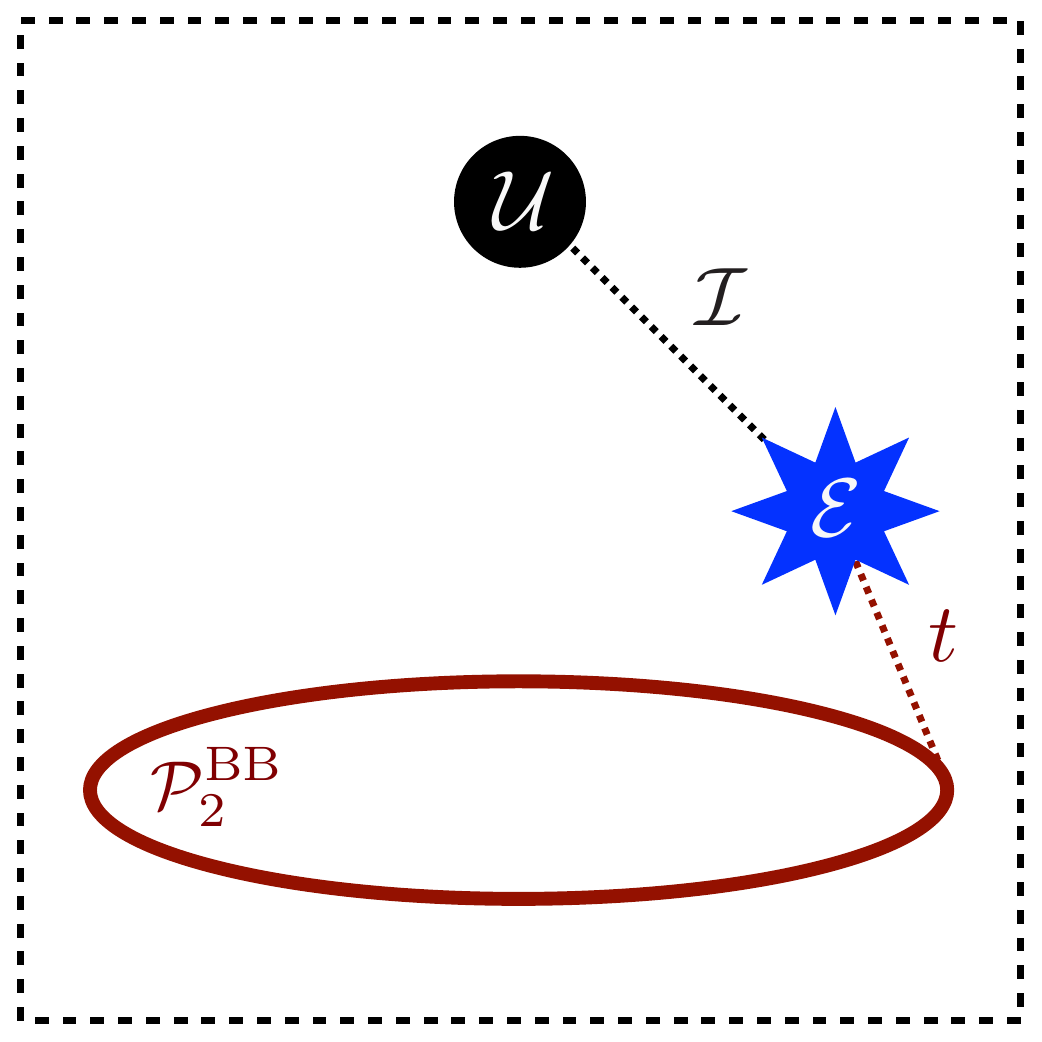}}\hspace{4mm}
\subfloat[]{\includegraphics[scale=0.685]{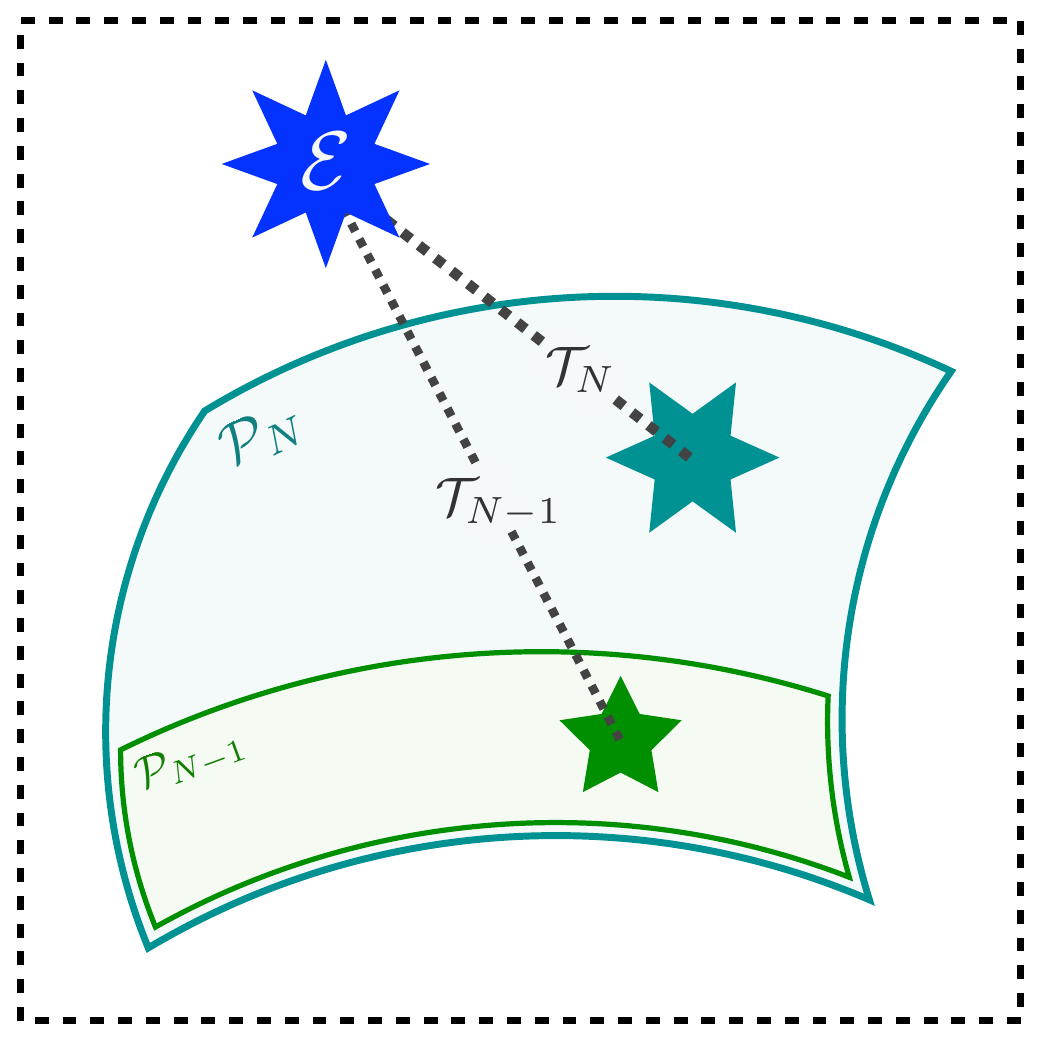}}
\caption{\label{fig:space_manyobs} An illustration of a variety of observables as distances between an event $\mathcal E$ and various manifolds in the space of events, as summarized in \Tab{tab:obs}.
(a) Thrust $t$ is the smallest distance from the event to the manifold $\mathcal P_2^\text{BB}$ of two-particle back-to-back events, while event isotropy $\mathcal I$ is the distance to the uniform event $\mathcal U$.
(b) $N$-jettiness observables $\mathcal T_N$ are the smallest distances from the event to the $N$-particle manifolds $\mathcal P_N$.
}
\end{figure}

For most of the observables in this section, the $R$ parameter is not needed, in which case we define a notion of EMD relevant for comparing events with equal energies:
\begin{equation}
\label{eq:emdRtoinfty}
\EMD_{\beta} (\mathcal E, \mathcal E') = \lim_{R \to \infty} R^\beta \, \EMD_{\beta,R} (\mathcal E, \mathcal E').
\end{equation}
This only has a finite limit if $\mathcal E$ and $\mathcal E'$ have the same total energy, which is a useful property to simplify our analysis.
Explicitly, when comparing events with equal energy, this EMD simplifies to:
\begin{equation}
\label{eq:emd_noR}
\EMD_{\beta} (\mathcal E, \mathcal E') = \min_{\{f_{ij}\ge0\}} \sum_{i=1}^M\sum_{j=1}^{M'} f_{ij} \theta_{ij}^\beta ,
\end{equation}
\begin{equation}
\label{eq:emdconstraints_noR}
\sum_{i=1}^M f_{ij} = E_j', \quad\quad \sum_{j=1}^{M'} f_{ij} = E_i, \quad\quad \sum_{i=1}^M\sum_{j=1}^{M'} f_{ij} = \sum_{i=1}^M E_i = \sum_{j=1}^{M'}E_j'.
\end{equation}
This will be the precise notion of EMD we use when the $R$ subscript is suppressed.

In \Tab{tab:obs}, we summarize some of the observables considered below and their geometric interpretations.
In \Fig{fig:space_manyobs}, we illustrate the geometric construction of many of these observables, which we will explore in detail below.

\subsection{Event-level observables}
\label{sec:eventobservables}

\subsubsection{Thrust}
\label{subsubsec:thrust}

Thrust is an observable that quantifies the degree to which an event is pencil-like~\cite{Brandt:1964sa,Farhi:1977sg,DeRujula:1978vmq}.
It has been experimentally measured~\cite{Barber:1979bj,Bartel:1979ut,Althoff:1983ew,Bender:1984fp,Abrams:1989ez,Li:1989sn,Decamp:1990nf,Braunschweig:1990yd,Abe:1994mf,Heister:2003aj,Abdallah:2003xz,Achard:2004sv,Abbiendi:2004qz} and theoretically calculated~\cite{Gehrmann-DeRidder:2007nzq,GehrmannDeRidder:2007hr,Becher:2008cf,Weinzierl:2009ms,Abbate:2010xh,Abbate:2012jh} in detail for electron-positron collisions.
Thrust seeks to find an axis $\hat n$ (the ``thrust axis'') such that most of the radiation lies in the direction of either $\hat n$ or $-\hat n$; i.e.~it maximizes the amount of radiation \emph{longitudinal} to the thrust axis.
While a variety of conventions for defining thrust exist, here we use the following dimensionful definition:
\begin{equation}\label{eq:thrust}
t(\mathcal E) = 2\min_{\hat n}\sum_{i=1}^M|\vec p_i|(1- |\vec n_i \cdot \hat n|),
\end{equation}
where $\hat n_i = \vec p_i/|\vec p_i|$ and other definitions follow by simple rescalings.
A thrust value of zero corresponds to an event consisting of two back-to-back prongs, while its maximum value of the total energy corresponds to a perfectly spherical event.

Interestingly, the value of thrust in \Eq{eq:thrust} is equivalent to the cost of an optimal transport problem.
This connection will allow us to cast thrust as a simple geometric quantity written in terms of the EMD.
Using $E_i=|\vec p_i|$ for massless particles and writing out the absolute value, we can cast \Eq{eq:thrust} as:
\begin{equation}\label{eq:thrust2}
t(\mathcal E) = 2\min_{\hat n}\sum_{i=1}^M E_i \min(1 - \hat n_i \cdot \hat n,\, 1 + \hat n_i \cdot \hat n).
\end{equation}
For a fixed $\hat n$, the summand in \Eq{eq:thrust2} is the transportation cost to move particle $i$ to the closer of $\hat n$ or $-\hat n$ with an angular measure of $\theta_{ij}^2 = 2n_i^\mu n_{j\mu}= 2 (1 - \hat n_i \cdot \hat n_j)$.
The sum is then the EMD between the event and a two-particle event consisting of back-to-back particles directed along $\hat n$, where the energy of each of the two particles is equal to the total energy in the corresponding hemisphere.
The minimization over $\hat n$ is equivalent to a minimization over all such two-particle events.

Thus, thrust is our first example of an observable that can be cast in the form of \Eq{eq:obsdefemd}.
First, we define the manifold of back-to-back two-particle events:
\begin{equation}
\label{eq:bbmanifold}
\mathcal P^{\rm BB}_2 = \left\{\left. \sum_{i=1}^2 E_i\, \delta(\hat n - \hat n_i)\,\, \right| \,\, E_i \ge 0, \,\, \hat{n}_1 = - \hat{n}_2 \right\}.
\end{equation}
Then, using the notation of \Eq{eq:emd_noR} with $\beta=2$,%
\footnote{As mentioned in footnote \ref{footnote:pWasser}, strictly speaking only the square root of $\text{EMD}_2$ is a proper metric.
Because the square root is a monotonic function, though, this has no impact on the interpretation of thrust as an optimal transport problem.}
thrust is the smallest EMD from the event to the $\mathcal P^{\rm BB}_2$ manifold:
\begin{equation}
\label{eq:thrustasEMD}
\begin{boxed}{
t(\mathcal E) = \min_{\mathcal E' \in \mathcal P^{\rm BB}_2} \text{EMD}_2(\mathcal E, \mathcal E'),}
\end{boxed}
\end{equation}
where the minimization is carried out over all back-to-back two-particle configurations.

Because of the $R \to \infty$ limit in \Eq{eq:emdRtoinfty}, the optimal back-to-back configuration is guaranteed to have the same total energy as the event $\mathcal E$, as desired.
Note that even if this analysis is carried out in the center-of-mass frame, the optimal back-to-back configuration will generically not be at rest, since it involves two massless particles with different energies.%
\footnote{We thank Samuel Alipour-fard for discussions related to this point.}
This suggests a possible variant of thrust where one restricts the two-particle manifold to only include events that are physically accessible, either by forcing $E_1 = E_2$ or by considering massive particles as in \App{sec:mass}.

\subsubsection{Spherocity}

Spherocity is an observable that also probes the jetty nature of events~\cite{Georgi:1977sf}.
It seeks to find an axis that minimizes the amount of radiation in the event \emph{transverse} to it according to the following criterion:
\begin{equation}
\label{eq:spherocity_orig}
s(\E) = \min_{\hat n} \left(\sum_{i=1}^ME_i|\vec n_i\times\hat n| \right)^2,
\end{equation}
where the original definition of spherocity is related to this by an overall rescaling.
In the small $s$ limit, where the event configurations are back to back, we can write $|\vec n_i\times\hat n|\simeq \sqrt{2(1- |\hat n_i\cdot\hat n|)}$ and obtain:
\begin{equation}
\label{eq:spherocity}
s(\mathcal E)  \simeq \min_{\hat n} \left( \sum_{i=1}^M E_i \sqrt{2(1 - |\hat n_i \cdot \hat n|)} \right)^2.
\end{equation}
We focus on this limiting form for the following discussion.

Similar to the case of thrust, we can identify the spherocity expression to be minimized as an optimal transport problem.
For a fixed $\hat n$, the summand in \Eq{eq:spherocity} is the cost to transport particle $i$ to the closer of $\hat n$ or $-\hat n$ with an angular measure of $\theta_{ij}=\sqrt{2n^\mu_i n_{j\mu}}$.%
\footnote{In fact, \Eq{eq:spherocity_orig} is already an optimal transport problem, using $\theta_{ij} = \sin \Omega_{ij}$, where $\Omega_{ij}$ is the opening angle between particles $i$ and $j$.  This has the same small angle behavior as $\theta_{ij} = 2 \sin \frac{\Omega_{ij}}{2}$ from \Eq{eq:theta_def}.}
The sum is once again the EMD from the event to the manifold of back-to-back events, with the minimization over $\hat n$ interpreted as a minimization over the manifold.

Spherocity, in the appropriate limit, is therefore the square of the smallest EMD (with $\beta=1$) from the event to the manifold $\mathcal P^{\rm BB}_2$ from \Eq{eq:bbmanifold}:
\begin{equation}\label{eq:spherocityasEMD}
\begin{boxed}{
\sqrt{s(\mathcal E)} = \min_{\mathcal E' \in \mathcal P^{\rm BB}_2} \text{EMD}_1(\mathcal E, \mathcal E').
}\end{boxed}
\end{equation}
Through this lens, spherocity differs from thrust (besides the overall exponent) solely in the angular weighting factor:  $\beta=1$ for spherocity and $\beta=2$ for thrust.
One could continue in this direction, defining the distance of closest approach for general $\beta$.
(This is related to the event shape angularities~\cite{Berger:2003iw}, with a key difference being that angularities are traditionally measured with respect to the thrust axis.)
Instead, we now turn towards enlarging the manifold itself.

\subsubsection{Broadening}

Recoil-free broadening~\cite{Larkoski:2014uqa} is an observable that is sensitive to two-pronged events that are not precisely back-to-back jets.
Here we focus on recoil-free broadening, to be distinguished from the original jet broadening~\cite{Rakow:1981qn,Ellis:1986ig,Catani:1992jc} which is defined in terms of the thrust axis.%
\footnote{There is an EMD-based definition of the original jet broadening, using the thrust axis defined by $\mathcal E_t = \argmin_{\mathcal E' \in \mathcal P^{\rm BB}_2} \text{EMD}_2(\mathcal E, \mathcal E')$.  With modified angular measure and normalization, the original jet broadening with respect to the thrust axis is $b_t(\mathcal E) = \text{EMD}_1(\mathcal E, \mathcal E_t)$.  Note the two different values of $\beta$ in these expressions.}
It differs from spherocity only in that it minimizes the same quantity over two ``kinked'' axes that need not be antipodal.
Though subtle, this difference gives rise to very important theoretical differences between broadening and spherocity in the treatment of soft recoil effects~\cite{Dokshitzer:1998kz}, as discussed extensively in~\Ref{Larkoski:2014uqa}.

Here, we use the following definition of broadening:
\begin{equation}\label{eq:broadening}
b(\mathcal E) = \min_{\hat n_L,\,\hat n_R} \sum_{i=1}^M E_i \min(\theta_{iL}, \theta_{iR}),
\end{equation}
where $\theta_{iL}$ and $\theta_{iR}$ are the angular distances between particle $i$ and $\hat n_L$ and $\hat n_R$, respectively.
The fact that $\hat n_L$ and $\hat n_R$ are minimized separately (rather than $\hat n_L = -\hat n_R$) is the key distinction between recoil-free broadening and previous observables.
For a fixed $\hat n_L$ and $\hat n_R$, the summand in \Eq{eq:broadening} is the cost to transport particle $i$ to the closer of $\hat n_L$ or $\hat n_R$ with an angular measure of $\theta_{ij} = \sqrt{2n_i^\mu n_{j\mu}}$.
The sum is then the EMD from the event to the manifold of all two-particle events, which need not be back-to-back, namely $\mathcal P_2$ from \Eq{eq:npmanifold}.
The minimization over $\hat n_L$ and $\hat n_R$ is then interpreted as a minimization over this manifold.

Thus, broadening is the smallest EMD with $\beta=1$ from the event to $\mathcal P_2$:
\begin{equation}\label{eq:broadeningasEMD}
\begin{boxed}{
b(\mathcal E) = \min_{\mathcal E' \in \mathcal P_2} \text{EMD}_1(\mathcal E, \mathcal E').
}\end{boxed}
\end{equation}
The geometrical formulation of broadening in \Eq{eq:broadeningasEMD} differs from that of spherocity in \Eq{eq:spherocityasEMD} only in that it does not restrict the manifold to back-to-back configurations.%
This distinction is important to extend these ideas beyond the two-particle manifold.

\subsubsection{$N$-jettiness}
\label{sec:njettiness}

$N$-jettiness~\cite{Stewart:2010tn} (see also \Ref{Brandt:1978zm}) is an observable that partitions an event into $N$ jet regions and, for hadronic collisions, a beam region.
Without a beam region, it is defined based on a minimization procedure over $N$ axes:
\begin{equation}\label{eq:Njettiness}
\mathcal T_N^{(\beta)} = \min_{\hat n_1,\cdots,\hat n_N} \sum_{i=1}^M E_i \min\left(\theta_{i1}^\beta, \theta_{i2}^\beta, \cdots, \theta_{iN}^\beta\right),
\end{equation}
where $\theta_{i1}$ through $\theta_{iN}$ are the angular distances between particle $i$ and axes $\hat n_1$ through $\hat n_N$, respectively.

We immediately identify the summand as the cost of transporting particle $i$ to the nearest axis.
For fixed $\hat n_1$ through $\hat n_N$, assigning the energy transported to each axis as the energy of that axis gives rise to an $N$-particle event.
The expression to be minimized is then the EMD between the original event and that $N$-particle event.
The minimization over $\hat n_1$ through $\hat n_N$ is interpreted as a minimization over all such $N$-particle events.

Therefore, $N$-jettiness is the smallest distance between the event and the manifold $\mathcal{P}_N$ of $N$-particle events.
Equivalently, one can view it as the EMD to the best $N$-particle approximation of the event, and we return to this interpretation in \Sec{subsec:xcone}.
Thus, we have:
\begin{equation}
\label{eq:Njettiness_asEMD}
\begin{boxed}{
\mathcal T_N^{(\beta)} = \min_{\mathcal E' \in \mathcal P_N} \text{EMD}_{\beta}(\mathcal E, \mathcal E').
}\end{boxed}
\end{equation}
We see that $N$-jettiness generalizes the geometric interpretation of broadening to a general $N$-particle manifold and a general angular weighting exponent $\beta$.

For hadronic collisions, initial state radiation and underlying event activity require the introduction of a ``beam'' (or out-of-jet) region~\cite{Stewart:2009yx,Stewart:2010tn,Berger:2010xi}.
This can be accomplished via the introduction of a beam distance $\theta_{i,\text{beam}}$ into the minimization of \Eq{eq:Njettiness}.
There are many possible beam measures~\cite{Jouttenus:2013hs,Stewart:2015waa}, including ones that involve optimizing over two beam axes $\hat n_a$ and $\hat n_b$.
For simplicity, we focus on $\theta_{i,\text{beam}} = R^\beta$ which makes no explicit reference to the beam directions~\cite{Thaler:2011gf}.
Dividing by an overall factor of $R^\beta$, this modified version of $N$-jettiness can be written as:
\begin{equation}\label{eq:Njettiness_with_beam}
\mathcal T_N^{(\beta,R)} = \min_{\hat n_1,\cdots,\hat n_N} \sum_{i=1}^M E_i \min\left(1, \frac{\theta_{i1}^\beta}{R^\beta}, \frac{\theta_{i2}^\beta}{R^\beta}, \cdots, \frac{\theta_{iN}^\beta}{R^\beta}\right).
\end{equation}
This definition of $N$-jettiness is similar to \Eq{eq:Njettiness}, though now a particle can be closer to the beam than to any axis.
In this case, we say that the particle is transported to the beam and removed for a cost $E_i$.
The summand is then the cost to transport the event to an $N$-particle event plus the cost of removing any particles beyond $R$ from any axes.

Remarkably, this precisely corresponds to the EMD when formulated for events of different total energy.
Namely, $N$-jettiness with this beam region is simply the smallest distance between the event and the manifold of $N$-particle events, with $R$ smaller than the radius of the space:
\begin{equation}
\label{eq:Njettiness_with_beam_asEMD}
\begin{boxed}{
\mathcal T_N^{(\beta,R)} = \min_{\mathcal E'\in \mathcal P_N} \text{EMD}_{\beta,R}(\mathcal E, \mathcal E').
}\end{boxed}
\end{equation}
Particles removed by the optimal transport procedure are interpreted as being part of the beam region.
This fact will also be relevant in \Sec{sec:seqrec} for understanding sequential recombination jet clustering algorithms as geometric constructions in the space of events.

\subsubsection{Event isotropy}
\label{sec:isotropy}

Our new geometric phrasing of these classic collider observables highlights the types of configurations that they are designed to probe.
Specifically, \Eq{eq:obsdefemd} can be interpreted as how similar an event is to the class of events on the manifold $\mathcal{M}$.
This framework also suggests regions of phase space that are poorly resolved by existing observables and provides a prescription for developing new observables by identifying new manifolds of interest.

Event isotropy~\cite{isotropytemp} is a recently-proposed observable that provides a clear example of this strategy.
It is based on the insight that distances from the $N$-particle manifolds (such as thrust and $N$-jettiness) are not well-suited for resolving isotropic events with uniform radiation patterns.
Having observables with sensitivity to isotropic events can, for instance, improve new physics searches for microscopic black holes or strongly-coupled scenarios.
This motivates event isotropy, which is the distance between the event $\mathcal E$ and an isotropic event $\mathcal U$ of the same total energy:
\begin{equation}\label{eq:isotropy}
\mathcal I^{(\beta)}(\mathcal E) =  \text{EMD}_{\beta}(\mathcal E, \mathcal{U}).
\end{equation}
Since $\mathcal E$ and $\mathcal{U}$ have the same total energy by construction, it is natural to normalize event isotropy by the total energy to make it dimensionless.
The analysis in \Ref{isotropytemp} focused primarily on $\beta = 2$, though this approach can be extended to a general angular exponent.
For practical applications, it is convenient to consider a manifold of quasi-isotropic events of the same total energy and then estimate event isotropy as the average EMD between an event and this manifold.

We can cast \Eq{eq:isotropy} into the form of \Eq{eq:obsdefemd} by introducing a manifold $\mathcal M_{\mathcal U}$ of uniform events with varying total energies:
\begin{equation}\label{eq:isotropy_emd}
\boxed{
\mathcal I^{(\beta)}(\mathcal E) = \min_{\mathcal{E}' \in \mathcal M_{\mathcal U}} \text{EMD}_{\beta}(\mathcal E, \mathcal{E}').
}
\end{equation}
The $R \to \infty$ limit in \Eq{eq:emdRtoinfty} enforces that the optimal isotropic approximation $\mathcal{U}$ has the same total energy as $\mathcal E$, as in the original event isotropy definition.

The particular notion of a uniform distribution depends on the collider context---spherical for electron-positron collisions and cylindrical or ring-like for hadronic collisions---with corresponding choices for the energy and angular measures.
The case of ring-like isotropy at a hadron collider is particularly interesting, since there are known simplifications for one-dimensional circular optimal transport problems.
For $\beta = 1$, ring-like event isotropy can be computed in $\mathcal O(M)$ runtime~\cite{DBLP:journals/jmiv/RabinDG11} and there are fast approximations for any $\beta\ge1$~\cite{DBLP:journals/jmiv/RabinDG11}.
This is much faster than the generic $\mathcal O(M^3 \log M)$ expectation for EMD computations, motivating further studies of these one-dimensional geometries.

\subsection{Jet substructure observables}
\label{sec:jetobservables}

\subsubsection{Jet angularities}

Jet angularities are the energy-weighted angular moments of radiation within a jet~\cite{Ellis:2010rwa} (see also \Refs{Almeida:2008yp,Larkoski:2014uqa,Larkoski:2014pca}).
Here, we use the following definition of a recoil-free jet angularity:
\begin{equation}
\lambda_\beta(\mathcal J) = \min_{\hat n} \sum_{i=1}^M E_i \,\theta_{i}^\beta,
\end{equation}
where $\theta_{i}$ is the angular distance between particle $i$ and an axis $\hat n$.
The summand of an angularity is the EMD from the jet to the axis, so we can follow the analogous logic from our previous discussions of event shapes to reframe this observable in our geometric language.
Specifically, the recoil-free angularities are the closest distance between the jet and the 1-particle manifold $\mathcal{P}_1$:
\begin{equation}
\label{eq:emd_jet_angularities}
\begin{boxed}{
\lambda_\beta(\mathcal J) = \min_{\mathcal J'\in\mathcal P_1} \text{EMD}_{\beta}(\mathcal J, \mathcal J').
}\end{boxed}
\end{equation}
One can alternatively consider a definition of angularities where $\theta_{i}$ is computed with respect to a fixed jet axis.
In that case, the angularities are the EMD from the jet to a 1-particle configuration where the total energy of the jet is placed at the position of the desired axis.

\subsubsection{$N$-subjettiness}
\label{subsec:nsubjettiness}

$N$-subjettiness is a jet substructure observable that applies the ideas of $N$-jettiness at the level of jet substructure~\cite{Thaler:2010tr,Thaler:2011gf}.
$N$ axes are placed within the jet, with a penalty for having energy far away from any axis, and then the positions of the axes are optimized.
The (dimensionful) $N$-subjettiness of a jet can be defined as follows:
\begin{equation}
\tau_N^{(\beta)} (\mathcal J) = \min_{\hat n_1,\cdots,\hat n_N} \sum_{i=1}^M E_i \min\left(\theta_{i1}^\beta, \theta_{i2}^\beta, \cdots, \theta_{iN}^\beta\right),
\end{equation}
where $\theta_{i1}$ through $\theta_{iN}$ are the angular distances between particle $i$ and axes $\hat n_1$ through $\hat n_N$.
The beam region is absent due to the fact that these observables are only defined using the particles already within an identified jet.

\begin{figure}[p]
\centering
\includegraphics[scale=0.7]{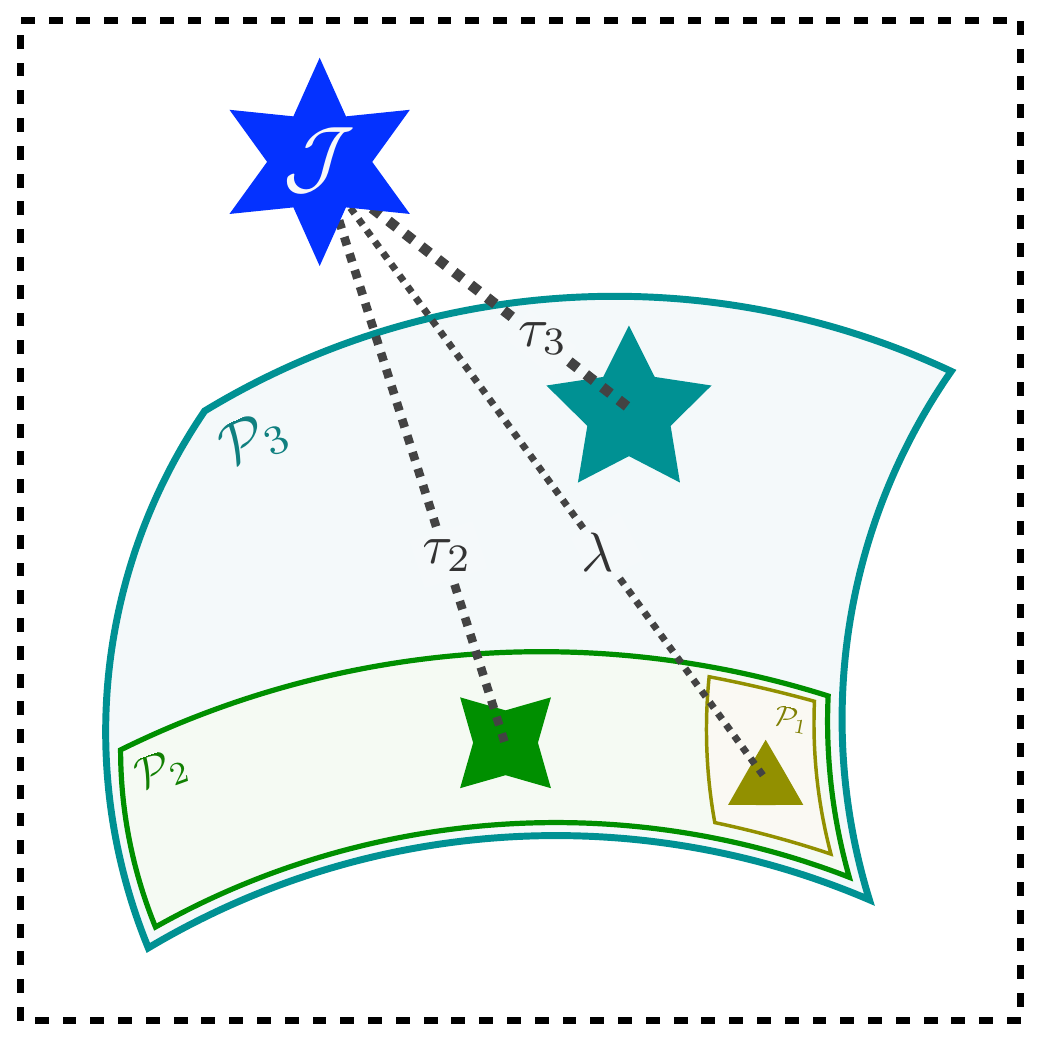}
\caption{\label{fig:space_nsub} An illustration of $N$-subjettiness values as the smallest distances, as measured by EMD, between the event $\mathcal E$ and each of the $N$-particle manifolds $\mathcal P_N$.
The jet angularities are the distances to the 1-particle manifold $\mathcal P_1$.
These observables form a set of ``coordinates'' for the space.
}
\end{figure}

We can find a geometric interpretation for $N$-subjettiness by using the analogous discussion from $N$-jettiness in \Sec{sec:njettiness}.
$N$-subjettiness is the distance between the jet and the manifold of all $N$-particle jets:
\begin{equation}
\begin{boxed}{
\tau_N^{(\beta)}(\mathcal J) = \min_{\mathcal J'\in\mathcal P_N} \text{EMD}_{\beta}(\mathcal J, \mathcal J').
}\end{boxed}
\end{equation}
As a limiting case, $N = 1$ corresponds to the jet angularities in \Eq{eq:emd_jet_angularities}.

In this way, we can view $N$-subjettiness values as ``coordinates'' for the space of jets, defined as distances from each of the $N$-particle manifolds, illustrated in \Fig{fig:space_nsub}.
The $N$-subjettiness ratios $\tau_{N} / \tau_{N-1}$, used ubiquitously for jet substructure studies~\cite{Larkoski:2017jix,Asquith:2018igt,Marzani:2019hun}, are then the relative distances between the manifolds $\mathcal{P}_N$ and $\mathcal{P}_{N-1}$.
This is also an interesting way to interpret existing constructions of observable bases using $N$-(sub)jettiness~\cite{Datta:2017rhs,Datta:2017lxt,Larkoski:2019nwj}; the fact that multiple $\beta$ values are typically needed for these constructions emphasizes that the choice of ground metric affects the geometry of the space induced by the EMD.

\begin{figure}[p]
\centering
\subfloat[]{\includegraphics[scale=0.685]{figures/space_xcone}}\hspace{4mm}
\subfloat[]{\includegraphics[scale=0.685]{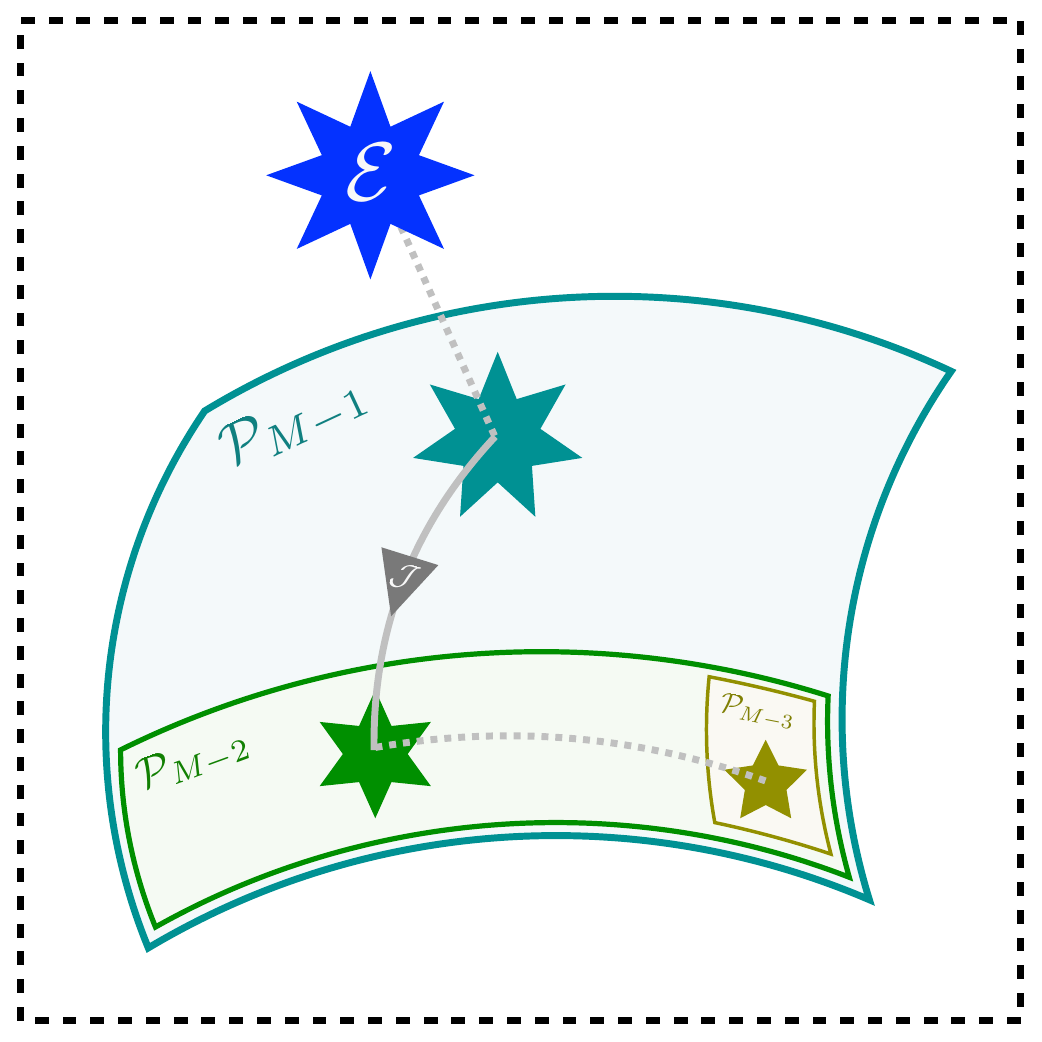}}
\caption{\label{fig:space_jets} An illustration of jet clustering algorithms as projections to $N$-particle manifolds $\mathcal P_N$ in the space of events.
(a) Exclusive cone finding algorithms yield $N$ jets as the closest $N$-particle approximation to the event, as measured by the EMD.
(b) Sequential recombination algorithms iteratively find the best $(M-1)$-particle approximation to the $M$-particle event, either (dashed) merging two particles or (solid) removing a particle and calling it a jet.}
\end{figure}

\pagebreak

\section{Jets: The closest $N$-particle description of an $M$-particle event}
\label{sec:jets}

In this section, we turn our attention to how jets are defined.
We interpret two of the most common classes of jet algorithms as simple geometric constructions in the space of events.
Intuitively, we find that jets are the best $N$-particle approximation to an $M$-particle event. 
Many existing techniques naturally emerge from this simple principle in fascinating ways.

First, we discuss exclusive cone finding, as this technique corresponds exactly to the intuition above that jets approximate the energy flow of an event using a smaller number of particles.
Next, we show that many sequential recombination algorithms can be derived by iteratively approximating an $M$-particle event using $M-1$ particles.
These jet-finding strategies are illustrated in \Fig{fig:space_jets} as projections to $N$-particle manifolds in the space of events.

\subsection{General $N$: Exclusive cone finding}
\label{subsec:xcone}

XCone~\cite{Stewart:2015waa,Thaler:2015xaa} is an exclusive cone finding algorithm that seeks to find jets by minimizing $N$-jettiness.
It returns a fixed number of jets based on the parameters $N$ and $R$, in the same spirit as the exclusive version of the $k_t$ sequential recombination algorithm~\cite{Catani:1993hr}.
XCone proceeds by finding the $N$ axes that minimize $N$-jettiness as defined in \Eq{eq:Njettiness_with_beam}:
\begin{equation}\label{eq:XCone}
\underset{\hat n_1,\cdots,\hat n_N}{\text{argmin}}\sum_{i=1}^M E_i \min\left(1, \frac{\theta_{i1}^\beta}{R^\beta}, \frac{\theta_{i2}^\beta}{R^\beta}, \cdots, \frac{\theta_{iN}^\beta}{R^\beta}\right).
\end{equation}
Together with the energy assigned to those axes, or equivalently the set of particles mapped to each axis, the $N$ axes from \Eq{eq:XCone} define $N$ jets.
The jet radius parameter $R$ controls which particles are not assigned to any jet (i.e.\ assigned to the beam region). 
Following the discussion in \Sec{sec:njettiness}, \Eq{eq:XCone} can be interpreted as finding the $N$-particle configuration that best approximates the event of interest.

In our geometric language, we can cast XCone as identifying the point of closest approach between an event $\mathcal E$ and the $N$-particle manifold $\mathcal{P}_N$:
\begin{equation}\label{eq:XConeasEMD}
\begin{boxed}{
\mathcal J^\text{XCone}_{N,\beta,R}(\mathcal E) = \underset{\mathcal J\in\mathcal P_N}{\text{argmin}}\,\,\text{EMD}_{\beta,R}(\mathcal E, \mathcal J).
}\end{boxed}
\end{equation}
Different variants of XCone correspond to different choices for the energy weight $E_i$ and the angular measure $\theta_{ij}$~\cite{Jouttenus:2013hs,Stewart:2015waa}, which in turn correspond to different choices for what defines the ``best'' $N$-particle approximation to an event.

As discussed in \Ref{Thaler:2015uja}, there is a close relationship between exclusive cone finding algorithms, stable cone algorithms~\cite{Blazey:2000qt,Ellis:2001aa,Salam:2007xv}, and jet maximization algorithms~\cite{Georgi:2014zwa,Ge:2014ova,Bai:2014qca,Bai:2015fka,Wei:2019rqy}.
For the choice of $\beta = 2$, the jet axis aligns with the jet momentum direction, which is known as the stable cone criterion~\cite{Blazey:2000qt,Ellis:2001aa}.
For $N = 1$, one can relate the optimization problem in \Eq{eq:XConeasEMD} to maximizing a ``jet function'' over all possible partitions of an event into one in-jet region and one out-of-jet region~\cite{Georgi:2014zwa}.
Iteratively applying the $N = 1$ procedure is related to the \text{SISCone} algorithm with progressive jet removal~\cite{Salam:2007xv}.
All of these various algorithms can now be interpreted in our geometric picture as different ways to ``project'' the event $\mathcal E$ onto the $N$-particle manifold $\mathcal{P}_N$.

\subsection{$N=M-1$: Sequential recombination}
\label{sec:seqrec}

Sequential recombination algorithms are a class of jet clustering algorithms that have seen tremendous use at colliders, particularly the anti-$k_t$ algorithm~\cite{Cacciari:2008gp} which is the current default jet algorithm at the LHC.
These methods utilize an interparticle distance $d_{ij}$, a particle-beam distance $d_{iB}$, and a recombination scheme for merging two particles.
The algorithm proceeds iteratively by finding the smallest distance, combining particle $i$ and $j$ if it is a $d_{ij}$, or calling $i$ a jet and removing it from further clustering if it is a $d_{iB}$.

There are a variety of distance measures and recombination schemes that appear in the literature, many of which are implemented in the \textsc{FastJet} library~\cite{Cacciari:2011ma}.
The most commonly used distance measures take the form:
\begin{equation}\label{eq:dij}
d_{ij} = \min\left(E_i^{2p}, E_j^{2p}\right) \frac{\theta_{ij}^2}{R^2},\quad \quad d_{iB} = E_i^{2p},
\end{equation}
where $p$ is an energy weighting exponent and $R$ is the jet radius.
The exponent $p=1$ corresponds to $k_t$ jet clustering~\cite{Catani:1993hr,Ellis:1993tq}, $p=0$ corresponds to Cambridge/Aachen (C/A) clustering~\cite{Dokshitzer:1997in,Wobisch:1998wt}, and $p=-1$ corresponds to anti-$k_t$ clustering~\cite{Cacciari:2008gp}.
The recombination scheme determines the energy $E_c$ and direction $\hat n_c$ of the combined particle and typically takes the form:
\begin{equation}
\label{eq:recomb}
E_c = E_i  + E_j,\quad \quad \hat n_c = \frac{E_i^\kappa \, \hat n_i + E_j^\kappa \, \hat n_j}{E_i^\kappa + E_j^\kappa},
\end{equation}
where $\kappa=1$ corresponds the $E$-scheme (most typically used), $\kappa=2$ is the $E^2$-scheme~\cite{Catani:1993hr,Butterworth:2002xg}, and $\kappa\to\infty$ is the winner-take-all scheme~\cite{Bertolini:2013iqa,Larkoski:2014uqa,Salambroadening}.
In the $E$-scheme, the four-momenta of the two particles are simply added.%
\footnote{\label{footnote:Escheme}One has to be a bit careful about the interpretation of jet masses in the $E$-scheme.  In the discussion below, the combined particle is interpreted as a massless four-vector.  For the angular distance in \Eq{eq:theta_def}, the direction $\hat{n}_i$ is the same for massless and massive particles, so one can consistently assign the mass of the jet to be the invariant mass of the summed jet constituents.  For the rapidity-azimuth distance typically used at hadron colliders, though, the rapidity of a particle depends on its mass, so one has to be careful about whether one is talking about a light-like jet axis or a massive jet when discussing the $E$-scheme.  See further discussion in \App{sec:mass}.}
In the winner-take-all scheme, the direction is determined by the more energetic particle.

\begin{table}[t]
\centering
\begin{tabular}{c|ccc|cc}
\hline\hline
EMD$_{\beta,R} $& Name & Measure $d_{ij}$ & $d_{iB}$ & Name & Scheme $\lambda^*$ \\ \hline\hline
$0<\beta < 1$ & Gen.\ $k_t$ & $\min(E_i, E_j) \frac{\theta_{ij}^\beta}{R^\beta}$ & $E_i$ & Winner-take-all  & $\text{argmin}(E_i,E_j)$ \\
$\beta = 1$ & $k_T$ & $\min(E_i, E_j) \frac{\theta_{ij}^{\phantom{\beta}}}{R}$ & $E_i$ & Winner-take-all  & $\text{argmin}(E_i,E_j)$ \\
$\beta = \frac32$ & ? & $\frac{E_i E_j}{\sqrt{E_i^2 + E_j^2}} \frac{\theta_{ij}^\frac32}{R^\frac32}$ & $E_i$ & $E^2$-scheme & $\frac{E_j^2 }{E_i^2 + E_j^2}$ \\
$\beta = 2$ & ? & $\frac{E_i E_j}{E_i + E_j} \frac{\theta_{ij}^2}{R^2}$ &  $E_i$ & $E$-scheme &  $\frac{E_j}{E_i + E_j}$ \\
\hline
$\beta\to\infty$ & C/A & $\frac{\theta^{\phantom{\beta}}_{ij}}{2R}$ & 1 & ? & $\frac12$\\ \hline \hline
\end{tabular}
\caption{Different sequential recombination measures $d_{ij}$ and recombination schemes $\lambda^*$ that emerge from an EMD formulation.
A question mark indicates a method that, to our knowledge, does not yet appear in the literature.
The traditional definitions of generalized $k_t$ and C/A require squaring $d_{ij}$ and $d_{iB}$.
Note the factor of 2 in the C/A effective jet radius parameter.
}
\label{tab:seqrec}
\end{table}

The conceptual and algorithmic richness of these different distance measures and recombination schemes arose from decades of phenomenological studies.
Remarkably, many of these techniques naturally emerge from event space geometry, as finding the point on the $(M-1)$-particle manifold $\mathcal P_{M-1}$ that is closest to configuration $\mathcal{E}$ with $M$ particles.
Note that the sequential recombination algorithms in \Eqs{eq:dij}{eq:recomb} depend on the two parameters $p$ and $\kappa$, whereas \Eq{eq:SRasEMD} depends only on $\beta$, so the logic below will only identify a one-dimensional family of jet algorithms, as summarized in \Tab{tab:seqrec}.

To derive this connection between event geometry and sequential recombination, we need the following simple yet profound lemma, using the suggestive notation of $d_{iB}$ and $d_{ij}$ to refer to the EMD cost of rearrangement.
\begin{lemma}
As measured by the EMD, the closest $(M-1)$-particle event to an $M$-particle event has, without loss of generality, either:
\begin{enumerate}
\item[(a)] Two of the particles in the event merged together.
\item[(b)] One of the particles in the event removed.
\end{enumerate}
\end{lemma}
\begin{proof}
Removing a particle from the event has some EMD cost $d_{iB}$ and merging a pair of particles has a some EMD cost $d_{ij}$.
To reduce the number of particles in the event by one, one can either remove a particle or merge two particles.
Altering more than two particles by (re)moving fractions of additional particles always incurs additional EMD costs.
If there are multiple pairs that are zero distance apart, then we can without loss of generality always choose to only merge one pair.
\end{proof}

The two options in this lemma correspond precisely to the two possible actions at each stage of a sequential recombination algorithm.
The EMD cost of removing a particle is always
\begin{equation}\label{eq:diBemd}
d_{iB} = E_i.
\end{equation}
If this is less than the cost of merging two particles together, then particle $i$ can be identified as a jet.
For one step of a sequential recombination (SR) procedure applied to an event $\mathcal E$ with $M$ particles, we can express this step mathematically as:
\begin{equation}\label{eq:SRasEMD}
\begin{boxed}{
\mathcal J_{\beta, R}^\text{SR}(\mathcal E) = \mathcal E - \underset{\mathcal E' \in \mathcal P_{M-1}}{\text{argmin}}\,\,\text{EMD}_{\beta,R}(\mathcal E, \mathcal E').
}\end{boxed}
\end{equation}
In our geometric picture, if the $M$ particle event is ``far away'' from the $(M-1)$-particle manifold $\mathcal P_{M-1}$, then the projected difference is a jet.

On the other hand, if the cost of merging two particles is less than any of the particle energies, then the event is ``close'' to the $(M-1)$-particle manifold.
Consider a pair of particles with energies $E_i$ and $E_j$ separated by a distance $\theta_{ij}$.
To find the best $(M-1)$-particle approximation, we want to merge these two particles into one combined particle with energy $E_i + E_j$.
Because the EMD is a metric, the optimal transportation plan must occur along a ``geodesic'' connecting the particles, with particle $i$ moving a distance $\lambda\,\theta_{ij}$ and particle $j$ moving a distance $(1-\lambda)\,\theta_{ij}$ for some $\lambda \in [0,1]$.\footnote{This linear decomposition of the distance does not hold for a general ground metric. However, it does hold when using the rapidity-azimuth distance, the opening-angle on the sphere, the small angle limit of \Eq{eq:theta_def}, or the improved $\theta_{ij}$ distance with particle masses in \App{sec:mass}.}
Minimizing this cost with respect to $\lambda$ yields both the cost of merging those two particles as well as the optimal recombination scheme with which to merge them.
Because no energy is removed in this process, \Eq{eq:SRasEMD} yields a zero energy jet, which we can interpret as no jet being found at this step of the sequential recombination.

The cost of merging particles $i$ and $j$ depends on the jet radius parameter $R$ and angular exponent $\beta$:
\begin{equation}\label{eq:emd2ps}
d_{ij} = \min_\lambda\left[E_i \left(\lambda \,\frac{\theta_{ij}}{R}\right)^\beta + E_j \left((1-\lambda) \,\frac{\theta_{ij}}{R}\right)^\beta\right].
\end{equation}
For $\beta \le 1$, the cost in \Eq{eq:emd2ps} is minimized at the endpoints.
This corresponds to moving the less energetic particle the entire distance $\theta_{ij}$ to the more energetic particle, which is the precisely behavior of the winner-take-all recombination scheme.
For $\beta>1$, the optimal value $\lambda^*$ can be found by differentiating \Eq{eq:emd2ps} with respect to $\lambda$ and setting the result equal to zero.
In general, the optimal recombination scheme has:
\begin{align}
0 < \beta \le 1:& \quad \lambda^* = 1 \text{ if $E_i < E_j$, else } 0, \nonumber \\
\beta>1: & \quad \lambda^* = \frac{1}{1 + \left(\frac{E_i}{E_j}\right)^{\frac{1}{\beta-1}}}.
\label{eq:opt_recomb}
\end{align}
To determine the actual cost, we substitute this $\lambda^*$ back into \Eq{eq:emd2ps}:
\begin{align}
\beta \le 1:&\quad d_{ij} = \min(E_i, E_j)\frac{\theta_{ij}^\beta}{R^\beta}, \nonumber\\
\beta>1: & \quad d_{ij} = \frac{E_i E_j^{\frac{\beta}{\beta-1} }+ E_i^{\frac{\beta}{\beta-1}}E_j}{\left(E_i^{\frac{1}{\beta-1}} + E_j^{\frac{1}{\beta-1}}\right)^\beta}\frac{\theta_{ij}^\beta}{R^\beta}.
\label{eq:dijemd}
\end{align}
If all $d_{ij}$ values in \Eq{eq:dijemd} are smaller than all particle energies in \Eq{eq:diBemd}, then the optimal transportation plan is to merge particles $i$ and $j$.

In this way, \Eq{eq:SRasEMD} takes an $M$-particle event and returns a jet (with zero energy if no actual jet is found) plus the remaining $(M-1)$-particle approximation.
This corresponds exactly to one step of a sequential clustering procedure.
Iterating this procedure until $M = 1$, we derive a sequential recombination jet algorithm, where the jets correspond to all of the positive energy configurations obtained from \Eq{eq:SRasEMD}.

Many existing methods reside within the simple framework of \Eq{eq:SRasEMD}.
For instance, $\beta=1$ corresponds to $k_t$ jet clustering with winner-take-all recombination.
The recombination scheme for $\beta=2$ is the $E$-scheme, whereas for $\beta=\frac32$ it is the $E^2$-scheme.
Raising the distance measures to the $1/\beta$ power and taking the $\beta\to\infty$ limit, we obtain the C/A clustering metric, albeit with an effective jet radius that is twice the $R$ parameter.
There are also a number of methods, indicated as question marks in \Tab{tab:seqrec}, that emerge from this reasoning yet do not presently appear in the literature.
Exploring these new methods is an interesting avenue for future work.

Intriguingly, in this geometric picture, the distance measure $d_{ij}$ and the recombination scheme $\lambda^*$ are paired by the $\beta$ parameter.
A similar pairing was noted in \Refs{Stewart:2015waa,Dasgupta:2015lxh} in the context of choosing approximate axes for computing $N$-(sub)jettiness, and it would be interesting to explore the phenomenological implications of these paired choices for jet clustering.
One sequential combination algorithm that does not appear is anti-$k_t$.
Given that anti-$k_t$ is a kind of hybrid between sequential recombination and cone algorithms, there may be a way to combine the logic of \Secs{subsec:xcone}{sec:seqrec} to find a geometric phrasing of anti-$k_t$.
If successful, such a geometric construction would likely illuminate the difference between exclusive jet algorithms like XCone that find a fixed number of jets $N$ and inclusive jet algorithms like anti-$k_t$ that determine $N$ dynamically.

\section{Pileup subtraction: Moving away from uniform events}
\label{sec:pileup}

The LHC era has brought with it new collider data analysis challenges.
One notable example is pileup mitigation~\cite{Soyez:2018opl}, removing the diffuse soft contamination from additional uncorrelated proton-proton collisions.
The radiation from pileup interactions is approximately uniform in the rapidity-azimuth plane, and several existing pileup mitigation strategies seek to remove this uniform distribution of energy from the event~\cite{Cacciari:2007fd,Krohn:2013lba,Cacciari:2014jta,Cacciari:2014gra,Bertolini:2014bba,Berta:2014eza,Komiske:2017ubm,Monk:2018clo,Martinez:2018fwc}.

\begin{figure}[t]
\centering
\includegraphics[width=0.5\columnwidth]{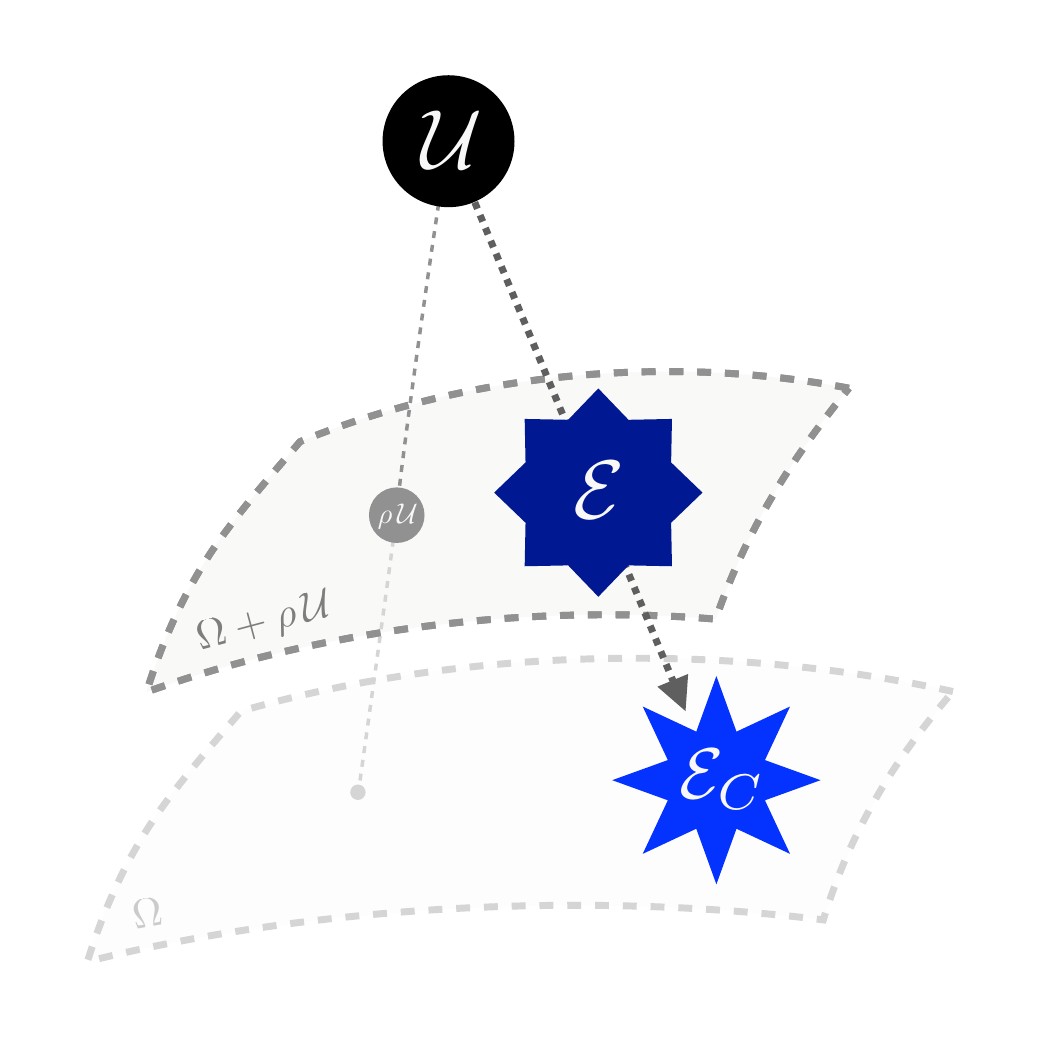}
\caption{
\label{fig:spacepu}
A visualization of pileup subtraction in the space of events as moving away from uniform radiation.
This proceeds by finding the event $\mathcal{E}_C$ that, when combined with uniform contamination $\rho \, \mathcal{U}$, is most similar to the given event $\mathcal{E}$.
Different pileup mitigation strategies implement this removal in different ways.
In the figure above, $\Omega$ refers to the space of all energy flows and $\Omega + \rho \, U$ is a subset of that space obtained by adding uniform contamination to every event configuration (shown as a separate manifold for ease of visualization). 
}
\end{figure}

In this section, inspired by the approximate uniformity of pileup, we consider a class of pileup removal procedures that can be described as ``subtracting'' a uniform distribution of energy with density $\rho$, denoted $\rho\,\mathcal U$, from a given event.
We take the pileup density per unit area $\rho$ to be given, for instance, by the area-median approach~\cite{Cacciari:2007fd}.
Given an event flow $\mathcal E$, the subtracted distribution $\mathcal E - \rho \, \mathcal{U}$ is typically not a valid energy flow, since the local energy density can go negative.
Therefore, to implement this principle at the level of energy distributions, we turn this logic around and declare the corrected event $\mathcal{E}_C$ to be one that is as close as possible to the given event $\E$ when uniform radiation $\rho\,\mathcal U$ is added to it:
\begin{equation}
\label{eq:pileupemd}
\E_C(\E, \rho)  = \argmin_{\E' \in \Omega}\,\EMD_\beta(\E,\E' + \rho\,\mathcal U).
\end{equation}
Here, $\Omega$ refers to the complete space of energy flows, and the $R\to\infty$ limit of the EMD from \Eq{eq:emdRtoinfty} enforces that the corrected distribution $\E_C$ has the correct total energy.

As illustrated in \Fig{fig:spacepu}, one can visualize \Eq{eq:pileupemd} as a procedure that subtracts a uniform component from the energy flow.
To make contact with existing techniques, we show that area-based Voronoi subtraction~\cite{Cacciari:2007fd,Cacciari:2008gn,Cacciari:2011ma} and ghost-based constituent subtraction~\cite{Berta:2014eza} can be cast in the form of \Eq{eq:pileupemd} in the low-pileup limit.
We then develop two new pileup mitigation techniques that have optimal transport interpretations even away from the low-pileup limit: Apollonius subtraction, which corresponds to exactly implementing \Eq{eq:pileupemd} for $\beta=1$, and iterated Voronoi subtraction, which repeatedly applies \Eq{eq:pileupemd} with an infinitesimal $\rho$.
Since pileup is characteristic of a hadron collider, throughout this section we compute the EMD using particle transverse momenta $p_{T,i}$ and rapidity-azimuth coordinates $\hat n_i = (y_i,\phi_i)$, with $\theta_{ij}$ being the rapidity-azimuth distance.
Typically, pileup is taken to be uniform in a bounded region of the plane (e.g.~$|y|< y_{\rm max}$), though the specifics will not significantly affect our analysis.
First, though, we establish an important lemma that justifies why the corrected distribution $\E_C$ has a particle-like interpretation.

\subsection{A property of semi-discrete optimal transport}

There is a direct connection between pileup subtraction in \Eq{eq:pileupemd} and semi-discrete optimal transport~\cite{hartmann2017semi}.
Semi-discrete means that we are comparing a discrete energy flow (i.e.~one composed of individual particles) to a smooth distribution (i.e.~uniform pileup contamination).

Importantly, if $\mathcal{E}$ is discrete, then the corrected distribution $\E_C$ will also be discrete.
This can be proved via the following lemma.
\begin{lemma}
\label{lemma:pileupemd}
$\E_C$ defined according to \Eq{eq:pileupemd} is strictly contained in $\E$, where containment here means that $\E-\E_C$ is a valid distribution with non-negative particle transverse momenta.
\end{lemma}
\begin{proof}
Suppose for the sake of contradiction that $\E_C$ is defined according to \Eq{eq:pileupemd} has some support where $\E$ does not.
Let $\tilde\E$ be the distribution that $\E_C$ flows to when $\E_C+\rho\,\mathcal U$ is optimally transported to $\E$, noting that by definition, $\tilde\E$ must be contained in $\E$.
By the linear sum structure of \Eq{eq:emd_noR}~\cite{OTbook}, we have the following relation:
\begin{equation}
\label{eq:emdineq}
\EMD_\beta(\E,\E_C+\rho\,\mathcal U)=\EMD_\beta(\tilde\E,\E_C)+\EMD_\beta(\E-\tilde\E,\rho\,\mathcal U).
\end{equation}
Now using the following property of $\EMD_\beta$ inherited from Wasserstein distances~\cite{hartmann2017semi}:
\begin{equation}
\label{eq:wassersteinprop}
\EMD_\beta(\mathcal E,\mathcal F) \ge  \EMD_\beta(\mathcal E+\mathcal G,\mathcal F+\mathcal G),
\end{equation}
with equality if $\beta=1$ and the ground metric is Euclidean, we add $\tilde\E$ to both arguments of the last term in \Eq{eq:emdineq} and apply \Eq{eq:wassersteinprop} to find:
\begin{equation}
\label{eq:emdineq2}
\EMD_\beta(\E,\E_C+\rho\,\mathcal U)\ge\EMD_\beta(\tilde\E,\E_C)+\EMD_\beta(\E,\tilde\E+\rho\,\mathcal U).
\end{equation}
Now using that $\EMD_\beta(\tilde\E,\E_C)>0$ by the assumption that they have different supports as well as the non-negativity of the EMD, we find:
\begin{equation}
\label{eq:contradiction}
\EMD_\beta(\E,\E_C+\rho\,\mathcal U)>\EMD_\beta(\E,\tilde\E+\rho\,\mathcal U),
\end{equation}
which contradicts the assumption that $\E_C$ is found according to \Eq{eq:pileupemd}.
Thus, we conclude that $\E_C$ has no support outside of the support of $\E$, verifying the claim.
\end{proof}

This lemma establishes that pileup mitigation strategies defined by \Eq{eq:pileupemd} act by scaling the energies of the particles in the original event $\mathcal{E}$, not by producing new particles.
Indeed, this is a desirable feature of many popular pileup mitigations schemes, including two well-known methods that we describe next.

\subsection{Voronoi area subtraction}
\label{sec:voronoi}

\begin{figure}[t]
\centering
\subfloat[Voronoi]{\includegraphics[width=0.33\columnwidth]{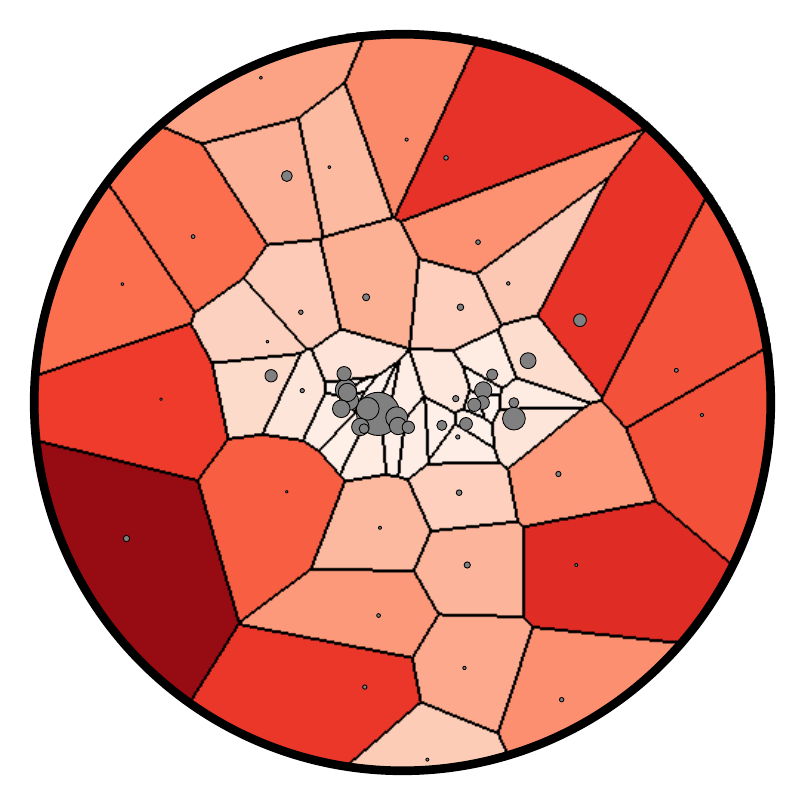}\label{fig:voronoi}}
\subfloat[Constituent Subtraction]{\includegraphics[width=0.33\columnwidth]{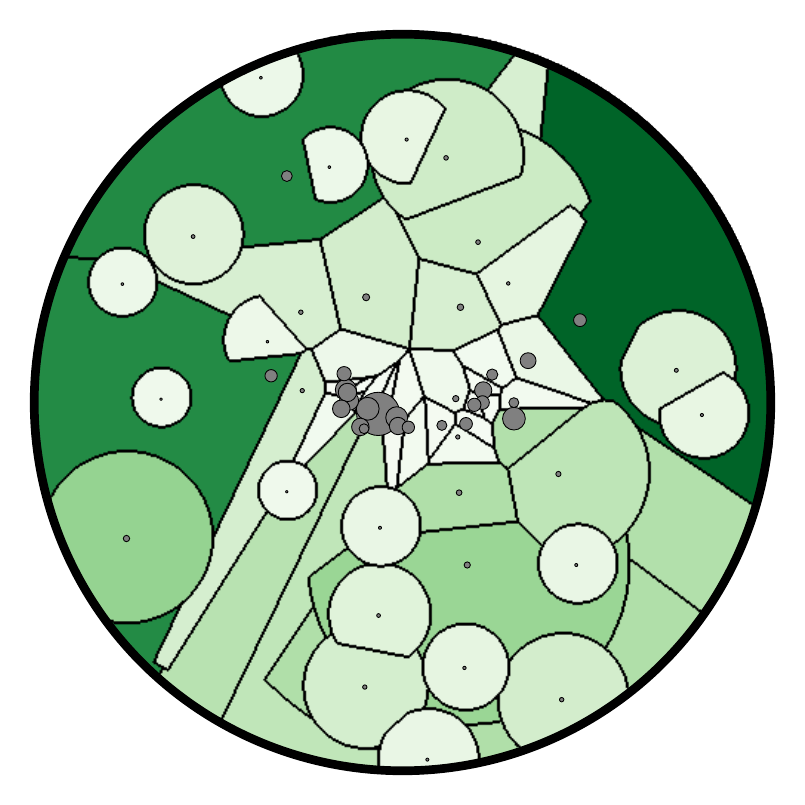}\label{fig:constsub}}
\subfloat[Apollonius]{\includegraphics[width=0.33\columnwidth]{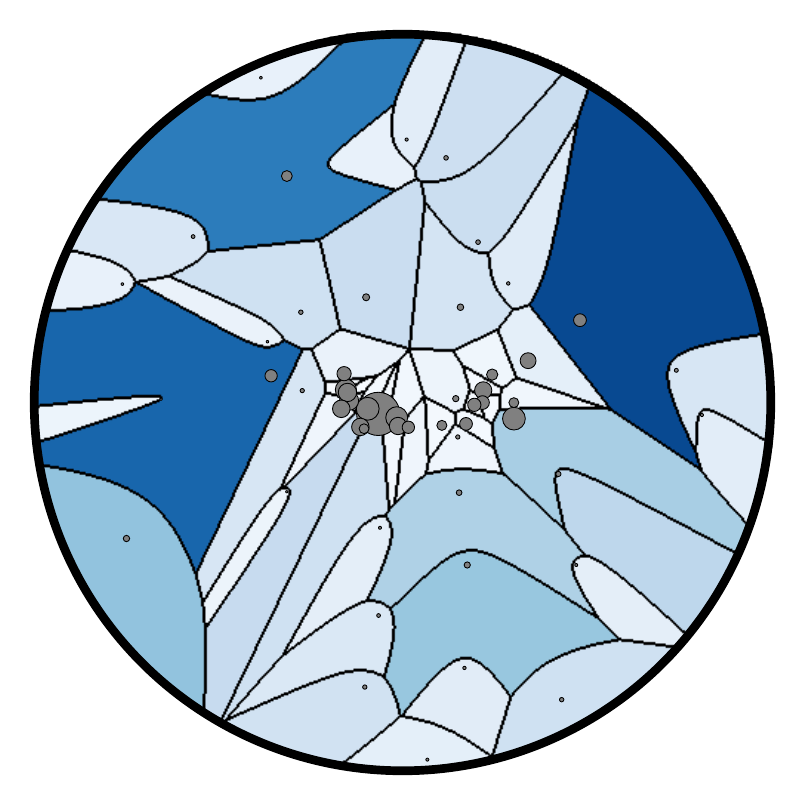}\label{fig:apollonius}}
\caption{
\label{fig:pileupexample}
Three pileup subtraction procedures shown for a jet from the 2011A CMS Jet Primary Dataset.
The jet constituents are shown as gray disks at their locations in the rapidity-azimuth plane with sizes proportional to their transverse momenta.
The boundary is at a distance $R=0.5$ from the jet axis at the center.
The color intensity of each region is proportional to its area, which determines the size of the pileup correction.
The Voronoi diagram (a) is independent of $\rho$.
The constituent subtraction (b) and Apollonius (c) diagrams are determined using a $\rho$ that corresponds to subtracting one-tenth of the total scalar $p_T$ of the jet.
}
\end{figure}

Voronoi area subtraction~\cite{Cacciari:2007fd,Cacciari:2008gn,Cacciari:2011ma} is a pileup mitigation technique that estimates a particle's pileup contamination by associating it with an area determined by its corresponding \emph{Voronoi region}, or the set of points in the plane closer to that particle than any other~\cite{Aurenhammer2013:book}.
Letting $A_i^\text{Vor.}$ be the area of the Voronoi region of particle $i$, Voronoi subtraction then simply removes $\rho A_i^\text{Vor.}$ from each particle's transverse momentum, without letting the particle $p_T$ become negative.
If $\rho A_i^\text{Vor.}\ge p_{T,i}$ then the particle is removed entirely.
In \Fig{fig:voronoi}, we show the Voronoi regions for an example jet recorded by the CMS detector~\cite{CMS:JetPrimary2011A,Komiske:2019jim}.

Voronoi area subtraction (VAS) can be thought of as carving up the uniform event $\rho\,\mathcal U$ according to the original event's Voronoi diagram and transporting this energy to the location of the corresponding particle, yielding the corrected energy flow:%
\begin{equation}
\label{eq:evas}
\E_\text{VAS}(\hat n)=\sum_{i=1}^M\max\left(p_{T,i}-\rho A_i^\text{Vor.},0\right)\delta(\hat n - \hat n_i).
\end{equation}
Strictly speaking, Voronoi area subtraction does not satisfy exact IRC invariance (see \Eqs{eq:exactirsafety}{eq:exactcsafety}) and thus it cannot in general be written as operating on energy flows.
The reason is that an exact IRC splitting changes the number of Voronoi regions as well as their areas.
In order for \Eq{eq:evas} to be valid, we therefore assume that particles with exactly zero transverse momentum are removed and exactly coincident particles are combined before applying the Voronoi area subtraction procedure.

In the limit that $\rho\le p_{T,i}/A_i^\text{Vor.}$ for all particles $i$, the max in \Eq{eq:evas} evaluates to just its first argument.
In this case, since no particle is assigned a larger correction than its own transverse momentum, the Voronoi diagram gives the optimal transportation plan that minimizes the \EMD of moving the uniform event with density $\rho$ onto the event of interest:
\begin{equation}\label{eq:evas2}
\boxed{
\rho \le \min\left\{p_{T,i}/A_i^\text{Vor.}\right\}: \quad \mathcal E_\text{VAS} = \argmin_{\mathcal E' \in \Omega} \EMD_1(\mathcal E, \mathcal E' + \rho\,\mathcal U).
}
\end{equation}
Thus, in this small-pileup limit, \Eq{eq:evas} agrees with \Eq{eq:pileupemd} with $\beta=1$.
Despite this attractive geometric interpretation, Voronoi area subtraction beyond this limit is sensitive to arbitrarily soft particles: the amount that is subtracted depends only on particle positions, through their Voronoi areas, and not their transverse momenta.

\subsection{Constituent subtraction}
\label{subsec:const_sub}

Constituent subtraction~\cite{Berta:2014eza} is another pileup mitigation method that resolves several pathologies of Voronoi area subtraction by correcting the particles in a manner that depends on both their positions and their transverse momenta.%
\footnote{In this discussion, we focus on the $\alpha=0$ case of constituent subtraction, as recommended by \Ref{Berta:2014eza}.}
This comes at the cost of requiring a fine grid of low energy ``ghost'' particles with $p_T^g=\rho A^\text{ghost}$, where $A^\text{ghost}$ is the area assigned to each ghost, as a proxy for the pileup contamination.
The algorithm is applied by considering the geometrically closest ghost-particle pair $k,i$ and modifying them via:
\begin{equation}
p_{T,i} \to \max(p_{T,i} - p_{T,k}^g, 0), \quad \quad p_{T,k}^g \to \max(p_{T,k}^g - p_{T,i}, 0),
\end{equation}
continuing until all such pairs have been considered.
Since the number of ghosts is typically large in order to have fine angular granularity, this iteration through all ghost-particle pairs can be computationally expensive.

Constituent subtraction (CS) in the continuum ghost limit can be geometrically described by placing circles around each point in the rapidity-azimuth plane and simultaneously increasing their radii.
Each point in the plane is assigned to the particle whose circle reaches it first.
Circles stop growing when $A_i^\text{CS}$, the area assigned to particle $i$, grows larger than $p_{T,i}/\rho$.
We can write the resulting distribution as:
\begin{equation}
\label{eq:ecs}
\E_\text{CS}(\hat n)=\sum_{i=1}^M\left(p_{T,i}-\rho A_i^\text{CS}\right)\delta(\hat n-\hat n_i).
\end{equation}
Unlike naive Voronoi area subtraction, continuum constituent subtraction satisfies exact IRC invariance, since a zero energy particle has zero $A^\text{CS}$ and an exact collinear splitting yields two areas that sum to the original $A^\text{CS}$.
Constituent subtraction is also better suited for intermediate values of $\rho$, where particles can be fully removed, since further corrections are distributed to the next closest particle instead of being ignored as in Voronoi area subtraction.

Due to the complicated shapes of the corresponding regions, it is difficult to describe the areas $A_i^\text{CS}$ analytically and in practice they need to be estimated using numerical ghosts.
An example of constituent subtraction is shown in \Fig{fig:constsub}, where it can be seen that some region boundaries are straight and thus contained in the Voronoi diagram of \Fig{fig:voronoi}.
Indeed, growing circles from a set of points and assigning points in the plane according to which circle reaches them first is another way of describing the construction of a Voronoi diagram.
Regions with circular boundaries correspond to softer particles that are fully subtracted by the constituent subtraction procedure.

When $\rho$ is sufficiently small such that no particle's region has a circular boundary (i.e.\ no circle stops growing), constituent subtraction is exactly equivalent to Voronoi area subtraction.
Constituent subtraction in the low-pileup limit is then also equivalent to optimally transporting the uniform event with density $\rho$ to the event of interest and subtracting accordingly, again in line with \Eq{eq:pileupemd} with $\beta=1$:
\begin{equation}\label{eq:ecs2}
\boxed{
\rho \le \min\left\{p_{T,i}/A_i^\text{Vor.}\right\}: \quad \mathcal E_\text{CS} = \argmin_{\mathcal E' \in \Omega} \EMD_1(\mathcal E, \mathcal E' + \rho\,\mathcal U).
}
\end{equation}
Constituent subtraction can also be extended with a $\Delta R^\text{max}$ parameter to restrict ghosts from affecting distant particles.
Our geometric formalism can also encompass this locality by re-introducing the $R$-parameter to the EMD in \Eq{eq:ecs2} with $R = \Delta R^\text{max}$.

\subsection{Apollonius subtraction}
\label{sec:apollonius}

Voronoi area subtraction and constituent subtraction both make contact with \Eq{eq:pileupemd} in the small-$\rho$ limit, but we would like to explore pileup subtraction based on optimal transport for all values of $\rho$.
By Lemma~\ref{lemma:pileupemd}, we know that the corrected event is contained in the original event, and by the decomposition properties of the EMD in \Eq{eq:emdineq}, we only need to consider the transport of $\rho\,\mathcal U$ to $\E$.
Since the total transverse momenta of $\rho\,\mathcal U$ and $\E$ are generally different, this is now an example of a semi-discrete, \emph{unbalanced} optimal transport problem~\cite{OTtheory,bourne2018semi}.

The problem of minimizing the EMD between a uniform distribution and an event is solved, for general $\beta$, by a \emph{generalized Laguerre diagram}~\cite{bourne2018semi}.
For the special case of $\beta=1$, which we focus on here, this is also known as the \emph{Apollonius diagram} (or additively weighted Voronoi diagram)~\cite{DBLP:conf/esa/KaravelasY02,DBLP:journals/comgeo/GeissKPR13,hartmann2017semi}, and for $\beta=2$ it is a \emph{power diagram}~\cite{DBLP:journals/tog/XinLCCYTW16}.
An Apollonius diagram in the plane is constructed from a set of points $\hat n_i$ that each carry a non-negative weight $w_i$ that is the $i^\text{th}$ component of a vector ${\bf w}\in\mathbb{R}_+^M$.
In the two-dimensional Euclidean plane, the Apollonius region associated to particle $i$ depending on ${\bf w}$ is:
\begin{equation}
\label{eq:apollonius}
R^\text{Apoll.}_i({\bf w})=\left\{\hat n\in\mathbb{R}^2\,\big|\,\|\hat n-\hat n_i\|-w_i\le\|\hat n-\hat n_j\|-w_j,\,\,\forall\,\,j\neq i\right\},
\end{equation}
where particle indices $i,j=1,\ldots,M$ and $\|\cdot\|$ is the Euclidean norm.
One interpretation of \Eq{eq:apollonius} is that region $i$ is all points closer to a circle of radius $w_i$ centered at $\hat n_i$ than to the corresponding circle for any other particle.
The boundaries of the Apollonius regions are contained in the set $\{\hat n\in\mathbb{R}^2\,|\,\|\hat n-\hat n_i\|-\|\hat n-\hat n_j\|=w_i-w_j\}$, which is a union of hyperbolic segments.
Note that adding the same constant to all of the weights does not change the resulting Apollonius diagram.
Hence, if all the weights are equal, they can equivalently be set to zero and we attain the Voronoi diagram as a limiting case of an Apollonius diagram.

We can now specify the action of Apollonius subtraction on an event using the areas of the Apollonius regions subject to the minimal EMD requirement:
\begin{align}
\label{eq:corrapollonius}
\E_\text{Apoll.}(\hat n) &=\sum_{i=1}^M\left(p_{T,i} - \rho A_i^\text{Apoll.}({\bf w}^*)\right)\delta(\hat n-\hat n_i),
\\ \label{eq:wstar}
{\bf w}^* &=\argmin_{{\bf w}\in\mathbb{R}_+^M}\sum_{i=1}^M\EMD_{1,R}\left(\{p_{T,i},\hat n_i\},\rho R_i^\text{Apoll.}(\bf w)\right),
\end{align}
treating $\rho R_i^\text{Apoll.}({\bf w})$ as an event with uniform energy density $\rho$ in that Apollonius region.
Here, \Eq{eq:corrapollonius} is analogous to \Eqs{eq:evas}{eq:ecs}, and \Eq{eq:wstar} implements the requirement that the EMD of the subtraction is minimal.
Note that the $R$ parameter in \Eq{eq:wstar} serves only to guarantee that it is more efficient to transport energy rather than create/destroy it.
As long as $2R$ is greater than the diameter of the space, $R$ has no impact on the solution other than to guarantee that $\rho A_i^\text{Apoll.}({\bf w}^*)$ does not exceed $p_{T,i}$, as this would be less efficient than transporting the excess energy elsewhere.
An example of an Apollonius diagram is shown in \Fig{fig:apollonius}, where hyperbolic boundaries of the Apollonius regions are clearly seen in the outer part of the jet and straight boundaries, matching those of the Voronoi diagram, are seen near the core.

In this way, Apollonius subtraction generalizes Voronoi area and constituent subtraction beyond the small-pileup limit, directly implementing \Eq{eq:pileupemd} for $\beta=1$ for all values of $\rho$:
\begin{equation}\label{eq:eapoll2}
\boxed{
\mathcal E_\text{Apoll.} = \argmin_{\mathcal E' \in \Omega} \EMD_1(\mathcal E, \mathcal E' + \rho\,\mathcal U).
}
\end{equation}
While the optimal solution in \Eq{eq:wstar} is based on an unbalanced optimal transport problem, the restatement in \Eq{eq:eapoll2} corresponds to balanced transport.
This same connection underpins Lemma~\ref{lemma:pileupemd}, guaranteeing that the corrected event in \Eq{eq:corrapollonius} involves the same $M$ directions as the original event, just with different weights.

To turn \Eq{eq:eapoll2} into a practical algorithm, we would need an efficient way to compute the weights according to \Eq{eq:wstar}.
While \Refs{OTtheory,bourne2018semi} have developed the theoretical framework of semi-discrete, unbalanced optimal transport needed to solve this convex minimization problem, they stop short of describing easily-implementable algorithms to attain practical solutions.
In order to create \Fig{fig:apollonius}, we were limited to using numerical ghosts to directly solve for the transport plan that minimizes the EMD cost of subtracting the uniform energy component from the event, which is too computationally costly for LHC applications.

If the target areas $A_i^\text{Apoll.}$ are previously specified, then the solution to \Eq{eq:wstar} simplifies~\cite{hartmann2017semi}.
Given that the areas depend nontrivially on the resulting weight vector, though, the only case where we know them ahead of time is when $\rho$ is such that \emph{all} of the energy will be exactly subtracted, in which case $A_i^\text{Apoll.}=p_{T,i}/\rho$.
Though this is not so useful for pileup, where we typically want to subtract an amount of energy less than the total, it does indicate that an Apollonius diagram can be found and used to compute the event isotropy from \Sec{sec:isotropy} without the use of numerical ghosts.
We leave the implementation of such a procedure to future work, though we note that \Ref{hartmann2017semi} has already built an implementation that relies on numerical ghosts to estimate the areas of the Apollonius regions rather than solving for them analytically.

\subsection{Iterated Voronoi subtraction}
\label{sec:ivs}

Given the difficulty of analytically solving \Eq{eq:wstar} and thus implementing Apollonius subtraction, we now develop an alternative method called iterated Voronoi subtraction that gives up a global notion of minimizing EMD but retains a local one.
In all three methods described above, the difficulty comes when a particle is removed in the course of subtracting pileup.
Otherwise, the above methods all reduce to subtracting transverse momentum according to the Voronoi areas of the regions corresponding to the particles, as in \Eq{eq:evas2}.
This suggests a procedure in which pileup is subtracted according to \Eq{eq:pileupemd} an infinitesimal amount at a time, thus ensuring that \Eq{eq:evas2} can be used at every stage of the procedure.

The area of the Voronoi cell of particle $i$ is now a function of the total amount of energy density that has been subtracted thus far, a quantity that starts at zero and will be integrated up to the target $\rho^\text{tot}$ over the course of the procedure.
When a particle loses all of its transverse momentum, it is removed from the Voronoi diagram and is considered to have zero area associated to it.
The removal of a particle from the diagram changes the Voronoi regions of all of its neighbors, and their areas are updated accordingly.
Denoting the area associated to particle $i$ after $\rho$ worth of energy density has been subtracted as $A_i^\text{IVS}(\rho)$, we can write the corrected distribution for iterated Voronoi subtraction (IVS) as:
\begin{equation}
\label{eq:corrivs}
\E_\text{IVS}(\hat n)=\sum_{i=1}^M\left(p_{T,i} - \int_0^{\rho^\text{tot}} d\rho\,A_i^\text{IVS}(\rho)\right)\delta(\hat n-\hat n_i).
\end{equation}

Unlike \Eqs{eq:corrapollonius}{eq:wstar}, \Eq{eq:corrivs} naturally lends itself to a simple and efficient implementation.
We can iteratively solve for $A_i^\text{IVS}(\rho)$ using the fact that the areas correspond to Voronoi regions, and furthermore that these regions change only when a particle is removed.
Let $\E^{(0)}$ be the initial event consisting of particles with transverse momenta $p_{T,i}^{(0)}$ in Voronoi regions with area $A_i^{(0)}$.
We subtract a total energy density $\rho^\text{tot}$ by breaking up the integral in \Eq{eq:corrivs} starting with $\rho^{(0)}=0$ and determining the boundaries from:
\begin{equation}
\label{eq:rhon}
\rho^{(n)}=\max\left\{\rho\,\Bigg|\,0\le\rho\le\rho^\text{tot}-\sum_{i=0}^{n-1}\rho^{(i)}\,\text{ s.t. }\rho A_j^{(n-1)}\le p_{T,j}^{(n-1)}\,\forall\,j\right\},
\end{equation}
where $n$ starts at 1 and goes up to at most $M$.
The values of $\rho^{(n)}$ can be expressed simply as:
\begin{equation}
\label{eq:rhon2}
\rho^{(n)}=\min\left(\min\left\{\frac{p_{T,j}^{(n-1)}}{A_j^{(n-1)}}\right\}, \rho^\text{tot} - \sum_{i=0}^{n-1} \rho^{(i)}\right),
\end{equation}
where the inner minimum is taken over all remaining particles with $p_{T,j}^{(n-1)}>0$.

The updated particle momenta from each piece of the integral in \Eq{eq:corrivs} are then:
\begin{equation}
p_{T,i}^{(n)}=p_{T,i}^{(n-1)}-\rho^{(n)}A_i^{(n-1)},
\end{equation}
and a particle is considered removed if its transverse momentum is zero, in which case it is also considered to have zero area.
The areas $A_i^{(n)}$ are determined by the Voronoi diagram of $\E^{(n)}$.
The above procedure terminates either when the total amount of energy density removed is equal to $\rho^\text{tot}$ or there are no more particles left in the event.
Thus, iterated Voronoi subtraction makes contact with the geometric perspective of \Eq{eq:pileupemd}, applying it in infinitesimal increments, resulting in the discrete steps:
\begin{equation}\label{eq:eitvor2}
\boxed{
\mathcal E_\text{IVS}^{(n)} = \argmin_{\mathcal E' \in \Omega} \EMD_1(\mathcal E^{(n-1)}, \mathcal E' + \rho^{(n)}\,\mathcal U).
}
\end{equation}
Said another way, this is simply a repeated application of Voronoi area subtraction:  subtract until a particle reaches zero momentum, and repeat until the desired energy density has been removed.

Iterated Voronoi subtraction is made even more attractive computationally when one considers that the Voronoi diagram of $\E^{(n)}$ does not need to be recomputed from scratch.
Rather, it can be obtained from the Voronoi diagram of $\E^{(n-1)}$ by removing a site and updating only the neighboring regions.
Thus, we only need to construct the Voronoi diagram of $\E^{(0)}$ and each removal can be done in constant (amortized) time as the average number of neighbors of any cell is no more than 6~\cite{Aurenhammer2013:book}.
We have constructed an implementation of iterated Voronoi subtraction that interfaces with \textsc{FastJet} and will explore its phenomenological properties in future work.

\section{Theory space}
\label{sec:theory}

When do two theories give rise to similar signatures?
In this section, we seek to generalize the intuition behind the EMD to obtain a metric between theories using their predicted cross sections in energy flow space.
A construction of such a distance and the induced ``theory space'' is conceptually useful and, in fact, naturally underpins several recently introduced techniques for collider physics.

We introduce the cross section mover's distance ($\Sigma$MD) as a metric for the space of theories.
Here, we treat a ``theory'' as an ensemble of event energy flows with corresponding cross sections, encompassing both the predictions of quantum field theories as well as the structure of collider datasets.
To accomplish this, we again make use of an EMD-like construction, except the $\Sigma$MD uses the EMD itself as the ``angles'' and the event cross sections as the ``energies'', as mentioned in \Tab{tab:emdsmdcomp}.
The resulting space of theories with the $\Sigma$MD as a metric is illustrated in \Fig{fig:theory_space}.
Interestingly, \Ref{Erdmenger:2020abc} also put a metric on theory space by using the Fisher information matrix.

\begin{figure}[t]
\centering
\includegraphics[scale=0.7]{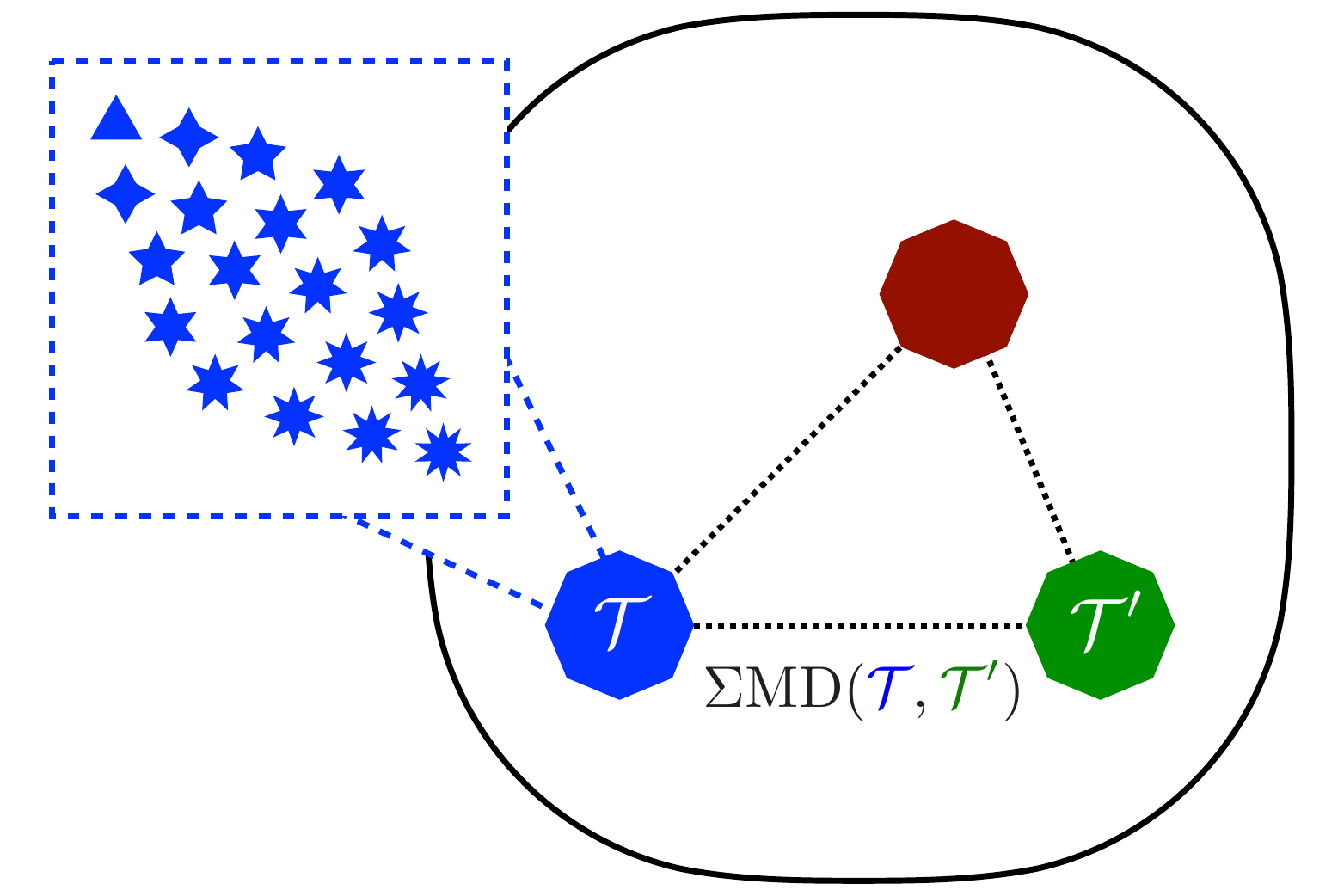}
\caption{\label{fig:theory_space} An illustration of the space of theories.
Each point in the space is a ``theory'': a distribution over (or collection of) events, as indicated by the blue point.
The distance between theories is quantified by the $\Sigma$MD, giving rise to a metric space.}
\end{figure}

\subsection{Introducing a distance between theories}

A ``theory'' $\mathcal T$ is taken to be a (finite, for now) set of events with associated cross sections $\{(\mathcal E_i,\sigma_i)\}_{i=1}^N$.
We can equivalently view $\mathcal T$ as a distribution over the space of event configurations $\mathcal E$:
\begin{equation}
\label{eq:T_def}
\mathcal{T}(\mathcal E) = \sum_{i = 1}^N \sigma_i \, \delta(\mathcal E - \mathcal E_i).
\end{equation}
In the case of unweighted events, the cross sections are simply $\sigma_i = 1/\mathcal{L}$, where $\mathcal{L}$ is the total integrated luminosity.
One can think of these $N$ events as being produced by a perfect event generator or measured by a perfect collider experiment in the context of this theory.
In the $\mathcal{L} \to \infty$ or $N\to\infty$ limit, \Eq{eq:T_def} becomes a smooth distribution.

The $\Sigma$MD is the minimum ``work'' required to rearrange one theory $\mathcal T$ into another $\mathcal T'$ by moving cross section $\mathcal F_{ij}$ from event $i$ in one theory to event $j$ in the other:
\begin{equation}
\label{eq:SigmaMD}
\text{$\Sigma$MD}_{\gamma,S; \beta, R}(\mathcal T, \mathcal T') = \min_{\{\mathcal F_{ij}\ge0\}} \sum_{i=1}^N\sum_{j=1}^{N'}  \mathcal F_{ij} \left(\frac{\text{EMD}_{\beta, R}(\mathcal E_i, \mathcal E_j')}{S} \right)^\gamma + \left|\sum_{i=1}^N \sigma_i - \sum_{j=1}^{N'}\sigma_j'\right|,
\end{equation}
\begin{equation}
\sum_{i=1}^N \mathcal F_{ij} \le \sigma_j', \quad \sum_{j=1}^{N'} \mathcal F_{ij} \le \sigma_i, \quad \sum_{i=1}^N\sum_{j=1}^{N'} \mathcal F_{ij} = \min\left(\sum_{i=1}^N \sigma_i,\sum_{j=1}^{N'}\sigma_j'\right).
\end{equation}
Here, $i$ and $j$ index the events in theories $\mathcal T$ and $\mathcal T'$, respectively. 
The parameter $S$, which has the same units as the EMD, controls the relative importance of the two terms, analogous to the jet radius parameter $R$ in the EMD.
We have also introduced a possible $\gamma$ exponent, analogous to the $\beta$ angular exponent in the EMD.

The $\Sigma$MD has dimensions of cross section, where the first term quantifies the difference in event distributions and the second term accounts for the creation or destruction of cross section.
For $\gamma > 0$, it is a true metric as long as the underlying EMD is a metric and $2S$ is larger than the largest attainable EMD between two events.%
\footnote{The analogous discussion to footnote \ref{footnote:pWasser} holds for $\gamma > 1$.}
In the limit as $S\to\infty$, the $\Sigma$MD reduces simply to the difference in total cross section between the two theories.
The natural continuum notion of \Eq{eq:SigmaMD} can be used whenever such an analysis is analytically tractable.

The $\Sigma$MD from a theory to itself is zero in the continuum limit or with infinite data.
Further, two theories that differ in their Lagrangians yet give rise to identical scattering cross sections for all energy flows will have a $\Sigma$MD of zero.
This includes, for instance, theories that are equivalent up to field redefinitions~\cite{Cheung:2017pzi} or rearrangements of the asymptotic states~\cite{Frye:2018xjj,Hannesdottir:2019rqq,Hannesdottir:2019opa}.
Finally, if two theories have any observable differences in energy flow, then the $\Sigma$MD between them will be non-zero.
Note that the $\Sigma$MD inherits the flavor and charge insensitivity of the EMD, but it is interesting to consider extending the $\Sigma$MD to account for additional quantum numbers that particles may carry.

\subsection{Jet quenching via quantile matching}

Quantile matching~\cite{Brewer:2018dfs} is an analysis strategy to study the modification of jets as they traverse the quark-gluon plasma in heavy-ion collisions.
We now show that, surprisingly, this technique can be cast naturally in a $\Sigma$MD formulation.
The optimal ``theory moving'' transport between two otherwise-equivalent datasets provides a proxy for the jet modification by the quark-gluon plasma.

Intuitively, the idea is to select a set of statistically equivalent jets from both proton-proton (pp) collisions and heavy-ion (AA) collisions.
This gives a snapshot of jets and their energies before and after modification by the quark-gluon plasma, respectively.
Such a selection can be achieved by selecting jets with the same upper cumulative effective cross section, after appropriately normalizing the AA cross section to account for the average number of nucleon-nucleon collisions.

With such a selection, a quantile matching can be used to specify $p_T^\text{quant}$ for a given heavy-ion jet with reconstructed transverse momentum $p_T^\text{AA}$:
\begin{equation}
\Sigma^\text{eff}_\text{pp}(p_T^\text{quant}) \equiv \Sigma_\text{AA}^\text{eff}(p_T^\text{AA}),
\end{equation}
where $p_T^\text{quant}$ gives a proxy for the jet $p_T$ prior to modification by the quark-gluon plasma.
The ratio between the heavy-ion and proton-proton jet transverse momentum in the same quantile then gives a physically-motivated quantification of the medium jet modification.
\begin{equation}
Q_\text{AA}(p_T^\text{quant}) = \frac{p_T^\text{AA}}{p_T^\text{quant}}.
\end{equation}

We now turn to explaining the intriguing connection between this quantile matching procedure and the $\Sigma$MD through optimal transport.
We use transverse momenta in place of energies and take $R\to\infty$ in \Eq{eq:emd}, where the EMD becomes simply the difference in jet transverse momenta.
Further, we set $\gamma = 1$ and $S=1$ in \Eq{eq:SigmaMD} and note that the normalization of the cross sections makes the second term in that equation vanish.

The theory moving problem now becomes a simple one-dimensional optimal transport problem of moving the pp jet $p_T$ distribution to the AA jet $p_T$ distribution.
Remarkably, this is mathematically equivalent to quantile matching.
We use the notation TM to represent the optimal theory movement $\mathcal F^*$ in the $\Sigma$MD.
Letting $\mathcal T_\text{AA}$ be the set of heavy-ion jets and $\mathcal T_\text{pp}$ be the set of proton-proton jets, we have that:
\begin{equation}
\begin{boxed}{
p_T^\text{quant} = \text{TM}\left(\mathcal T_\text{AA}, \mathcal T_\text{pp}\right)[p_{T}^\text{AA}],
}\end{boxed}
\end{equation}
where we can define this formally using a ``ghost'' heavy-ion jet with transverse momentum $p_T^\text{AA}$ and infinitesimal cross section $\sigma \sim 0$.

Quantile matching can therefore be seen as a matching induced by the optimal theory movement between the heavy-ion and proton-proton jets.
In this sense, it operationally defines the modification by the quark-gluon plasma in terms of the theory-movement of the jet transverse momentum spectrum.
It would be interesting to follow this connection further and explore this procedure using the full EMD beyond the $R\to\infty$ limit to study the medium modification as a function of the jet substructure.

\subsection{Event clustering and coresets}

One of the essential unsupervised methods for probing a dataset is to analyze its complexity.
A method to do this for collider physics datasets is that of $k$-eventiness, recently introduced in~\Ref{Komiske:2019jim}.
Here, one seeks to find $k$ representative events that minimize the EMD from each event in the dataset to the nearest representative event:
\begin{equation}
\mathcal V^{(\gamma)}_k = \min_{\mathcal K_1, \ldots,\mathcal K_k}\sum_{i=1}^N \sigma_i \min\{\text{EMD}(\mathcal E_i, \mathcal K_1)^\gamma, \text{EMD}(\mathcal E_i, \mathcal K_2)^\gamma, \ldots, \text{EMD}(\mathcal E_i, \mathcal K_k)^\gamma\},
\end{equation}
where we have dropped the $\beta$ and $R$ subscripts on \text{EMD} for compactness.
The value of $\mathcal V_k$ probes how well the dataset is approximated by the $k$ events.
This gives rise to the notion of $\mathcal V_k$ as the ``$k$-eventiness'' of the dataset, in analogy with $N$-(sub)jettiness, where smaller values of $\mathcal V_k$ indicate better approximations.

From a geometric perspective, $\mathcal V_k$ is the smallest $\Sigma$MD to the manifold of $k$-event datasets.
Analogous to \Eq{eq:emd_noR}, we can introduce the $S \to \infty$ version of the $\Sigma$MD:
\begin{equation}
\label{eq:smd_noR}
\text{$\Sigma$MD}_{\gamma} (\mathcal T, \mathcal T') = \lim_{S \to \infty} S^\gamma \, \text{$\Sigma$MD}_{\gamma,S} (\mathcal T, \mathcal T'),
\end{equation}
which yields an infinite distances between theories of different total cross sections.
Following the identical logic to \Secs{sec:njettiness}{subsec:nsubjettiness}, we have:
\begin{equation}
\label{eq:keventiness_EMD}
\begin{boxed}{
\mathcal V^{(\gamma)}_k = \min_{|\mathcal T'|=k} \text{$\Sigma$MD}_\gamma(\mathcal T, \mathcal T').
}\end{boxed}
\end{equation}
Here, we use the $|\cdot|$ notation to count the number of events in $\mathcal T'$.
Just like for $N$-(sub)jettiness, different values of $\gamma$ highlight different aspects of theory space geometry.

Following the logic in \Sec{subsec:xcone} of lifting the $N$-jettiness observable into the XCone jet algorithm, we can lift $k$-eventiness into an event clustering algorithm.
The representative events $\mathcal K$ (i.e.\ the point of closest approach on the $k$-event manifold), has the interpretation of the ``$k$-geometric-medians'' for $\gamma = 1$ or ``$k$-means'' for $\gamma = 2$.
For practical applications, it is often convenient to restrict the representative events to be within the dataset $\mathcal T$, i.e. the ``$k$-medoids'', giving only an approximate value of $\mathcal V_k$.
While the full problem of finding the representative jets may be computationally intractable, fast approximations to find the medoids exist and have been explored in \Ref{Komiske:2019jim}.

Inspired by \Sec{sec:seqrec}, one might consider implementing sequential clustering algorithms by iterating $\Sigma$MD computations to approximate $M$ events with $M-1$ events and so forth.
Such a clustering may be helpful for rigorous data compression of large collider datasets or, if implemented efficiently, for tasks such as triggering.
These ideas are closely related to the notion of finding a coreset (see \Ref{2019arXiv191008707J} for a recent review), for which techniques from quantum information and quantum computation may also find use~\cite{harrow2020small}.
Additionally, \Ref{DBLP:journals/corr/abs-1805-07412} uses the Wasserstein metric to construct ``measure coresets'' that take into account the underlying data distribution and which may prove useful for high-energy physics applications.
We leave further exploration of theory geometry and theory space algorithms to future work.

\section{Conclusions}
\label{sec:conc}

In this paper, we have explored the metric space of collider events from a theoretical perspective.
Beginning from the EMD between final states, namely the ``work'' required to rearrange one into another, we have cast a multitude of diverse collider algorithms and analysis techniques in a geometric language.
First, we connected this metric to the fundamental notion of IRC safety in massless quantum field theories, with the EMD providing a sharp language to define IRC safety and even Sudakov safety.
We extended this connection by highlighting that a wide variety of collider observables, including thrust and $N$-jettiness, can be cast as distances between events and manifolds in this space.
Further, we demonstrated that many jet clustering algorithms, such as exclusive cone finding and sequential clustering, can be exactly derived in full detail from the simple principle that jets are the best $N$-particle approximation to the event.
Even pileup mitigation techniques developed to face the LHC-era challenge of high luminosity running can be cast in the language of subtracting a uniform radiation pattern, connecting this field to semidiscrete unbalanced optimal transport.
Finally, we generalized our reasoning to define a distance between ``theories'' as sets of events with cross sections, proving a new lens to understand several existing techniques and a roadmap for future developments in the geometry of theory space.

From the perspective of massless quantum field theories, our metric space of events is the natural space for understanding observables, as the only truly observable quantities are IRC safe.
More speculatively, it would be interesting to circumvent the (unphysical) particle-level stage of calculations and make theoretical predictions directly in the space of events.
Understanding and expanding in this direction would require natural notions of volume and integration in this space, perhaps aided by recent developments in Wasserstein spaces~\cite{DBLP:journals/siamma/BianchiniB10,DBLP:journals/siamma/AguehC11,DBLP:conf/gsi/BertrandK13,DBLP:conf/gsi/GouicL15}, though we leave this fascinating exploration to future work.
Nonetheless, it has already been established that the energy flows themselves obey factorization theorems in effective field theory contexts~\cite{Bauer:2008jx}, and give rise to rich behavior in correlators~\cite{Basham:1978zq,Belitsky:2013ofa,Dixon:2019uzg,Chen:2019bpb}.
Going directly from first principles and symmetries to observables (i.e.\ energy flows, not particles) suggests a natural extension to the philosophy driving the present scattering amplitudes program (see \Refs{Elvang:2013cua,Carrasco:2015iwa,Cheung:2017pzi} for reviews).
It is also interesting to extend this logic to massive quantum field theories where observable quantities can depend on flavor and charge.
One promising strategy is to treat events as collections of objects (jets, electrons, muons, etc.) and to use a ground distance that penalizes converting one type of object into another~\cite{Romao:2020ojy}.
Alternatively, it may be possible to find an EMD variant that mimics the flavor-sensitive clustering behavior of \Ref{Banfi:2006hf}.

It is also useful to discuss these developments in the broader context of machine learning and the physical sciences~\cite{Larkoski:2017jix,Radovic:NatureML,Carleo:2019ptp}.
Typically, problems in the natural sciences can be cast as machine learning problems such as classification and regression, whereby the relevant tools from machine learning can be applied to achieve improved performance on those tasks.
It is far rarer for machine learning to enhance our theoretical or conceptual understanding of physics directly.
This story provides an interesting case where new insights and questions exposed by machine learning have impacted purely theoretical and phenomenological collider physics.
The question of when two collider events are similar, for which the EMD was introduced, was originally motivated by unsupervised learning methods and autoencoders~\cite{Hajer:2018kqm,Heimel:2018mkt,Farina:2018fyg,Roy:2019jae}, which require a distance matrix or reconstruction loss.
By providing an answer to this simple question, which itself involved familiar machine learning tools such as optimal transport, we uncovered a new mathematical formalism to better understand and express concepts in quantum field theory and collider physics.
We hope that this will be just one example of many future profound insights into the natural sciences facilitated by this perspective.

\acknowledgments

We are grateful to Samuel Alipour-fard, Jasmine Brewer, Cari Cesarotti, Aram Harrow, Andrew Larkoski, Simone Marzani, David Miller, Benjamin Nachman, Miruna Oprescu, Gavin Salam, Gregory Soyez, and Frank Tackmann for helpful conversations.
This work was supported by the Office of Nuclear Physics of the U.S. Department of Energy (DOE) under grant DE-SC-0011090 and by the DOE Office of High Energy Physics under grants DE-SC0012567 and DE-SC0019128.
JT was additionally supported by the Simons Foundation through a Simons Fellowship in Theoretical Physics.
We acknowledge the importance of minimizing transportation plans in these challenging times.

\appendix

\section{Energy moving with massive particles}
\label{sec:mass}

In this appendix, we briefly explore an alternative definition of energy flow appropriate for massive particles, with a corresponding change in the measures used to define the EMD.
The energy flow in \Eq{eq:energyflow} treats events as sets of particles that have energy-like weights $\{E_i\}$ and geometric directions $\{\hat{n}_i\}$.
The EMD in \Eq{eq:emd} is based on pairwise distances $\{\theta_{ij}\}$ that are only functions of the $\hat{n}_i$ and $\hat{n}_j$ directions.
The exact definitions of $E_i$, $\hat{n}_i$, and $\theta_{ij}$ may vary depending on the collider context and other choices.
For massless final-state particles in $e^+e^-$ collisions, it is typical to take the energy $E$ to be equal to the total momentum $|\vec p\,|$, and the geometric direction $\hat{n}$ to be equal to the unit vector $\vec{p} / E$.
For massless particles in hadronic collisions, it is natural to use transverse momentum $p_T$ and a geometric direction based on azimuth $\phi$ and pseudorapidity $\eta$.

It is straightforward to adapt the energy flow to massive particles (see related discussion in \Ref{Mateu:2012nk}).
For the energy measure, the natural choices are energy in the $e^+e^-$ case and transverse energy in the hadronic case:
\begin{equation}
E_i = \sqrt{|\vec p_i|^2+m_i^2}, \qquad
E_{Ti}=\sqrt{p_{Ti}^2+m_i^2}.
\end{equation}
Both of these reduce nicely to the expected expressions in the $m_i\to0$ limit.
For the geometric direction, the natural choices are velocity and transverse velocity, written in four-vector notation:
\begin{equation}
n^\mu_i = \frac{p_i^\mu}{E_i} = (1,\vec{v})^\mu, \qquad
n^\mu_{Ti} = \frac{p_i^\mu}{E_{Ti}} = (\cosh y_i, \vec{v}_{Ti}, \sinh y_i)^\mu,
\end{equation}
where $\vec{v} = \vec{p}_i / E_i$ is the particle three-velocity, $\vec{v}_{Ti} = \vec{p}_{Ti} / E_{Ti}$ is the particle transverse two-velocity, and $y_i$ is the particle rapidity.
Again, these have the expected behavior in the $m_i\to0$ limit, and for finite mass, the velocities are bounded as $|\vec{v}| \in [0,1]$ and $|\vec{v}_T| \in [0,1]$.

To define the EMD, we choose the following pairwise angular distance:
\begin{equation}
\label{eq:theta_massive}
\theta_{ij}=\sqrt{ - (n_i^\mu-n_j^\mu)^2},
\end{equation}
where one replaces $n^\mu$ with $n_T^\mu$ in the hadronic case.
The first minus sign is needed because the difference between two time-like vectors with $n^2 \in [0,1]$ is space-like.
This expression reduces to the usual expression $\theta_{ij}=\sqrt{2n_i^\mu n_{j\mu}}$ in the massless limit.

\begin{figure}[t]
\centering
\includegraphics[scale=0.7]{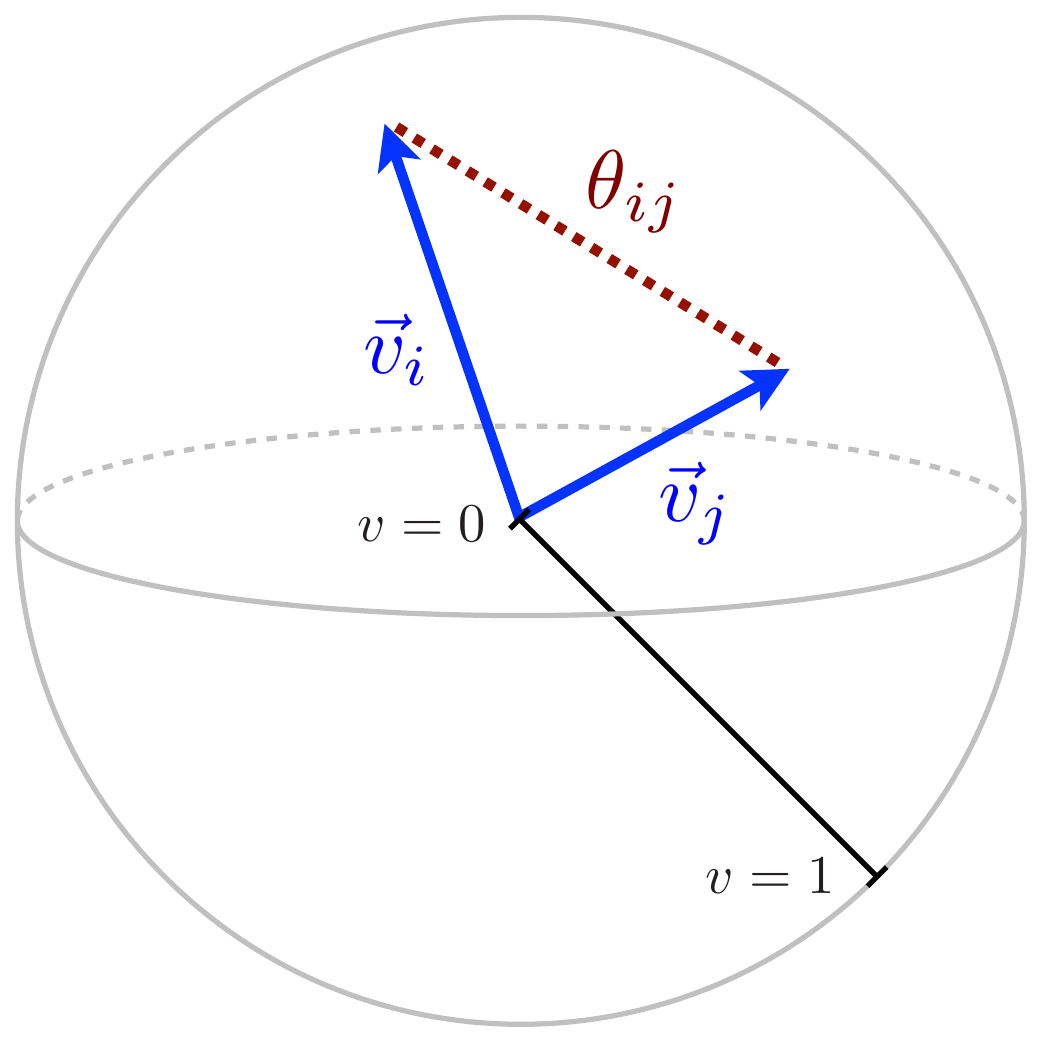}
\caption{\label{fig:sphere} The space of (massive) particle kinematics, with pairwise distances corresponding to Euclidean distances in this space.
Massless particles with $v=1$ live on the boundary and particles at rest with $v=0$ are at the origin.
One can interpret this figure as a snapshot of the event taken at a time $t$ after the collision, when the particles have traveled a distance $v t$.
}
\end{figure}

To gain intuition for this geometric distance between massive particles, it is instructive to expand out \Eq{eq:theta_massive} in the $e^+e^-$ case:
\begin{equation}
\theta_{ij}=\sqrt{\left(\vec v_i-\vec v_j\right)^2}=\sqrt{v_i^2+v_j^2 - 2 \, v_i \, v_j \cos \Omega_{ij}},
\end{equation}
where $\Omega_{ij}=\arccos \hat n_i\cdot \hat n_j$ is the purely geometric angle between particles $i$ and $j$.
We see that the velocity magnitude $v = |\vec{v}|$ acts as a radial coordinate on the sphere, and the pairwise distances $\theta_{ij}$ are just the Euclidean distances between two points in the unit ball, with distances $v_i$ and $v_j$ from the origin and angle $\Omega_{ij}$ between them.
Massless particles live entirely on the boundary with $v=1$ and massive particles live inside the ball with $0\le v<1$.
An illustration of this massive particle phase space is shown in \Fig{fig:sphere}.

The use of this massive distance measure has an interesting interplay with some of the studies in the body of the paper.
For example, the analysis of thrust in \Sec{subsubsec:thrust} involved finding the EMD to the manifold of back-to-back massless particle configurations of potentially unequal energy.
Using the massive particle distance, one could consider finding the EMD to the manifold of all possible two-particle configurations, including massive particles.
For $\beta = 2$, this is equivalent to partitioning the event into two halves with masses $M_A$ and $M_B$ and corresponding energies $E_A$ and $E_B$, and minimizing the quantity $M_A^2 / E_A + M_B^2 / E_B$.
A nice feature of this approach is that the optimal two particle configuration has the same energies and velocities as one would get from clustering the particles in each half.
Note that this approach is closely related to (but not identical to) the original definition of heavy jet mass in \Ref{Clavelli:1981yh} based on minimizing $M_1^2 + M_2^2$.

The idea of optimizing jet regions based on $M^2 / E$ also appears in the jet maximization approach~\cite{Georgi:2014zwa}.
In fact, using the massive distance measure in \Eq{eq:XConeasEMD} with $\beta = 2$ and $N=1$, and repeating the logic in \Ref{Thaler:2015uja}, we recover precisely the algorithm in \Ref{Georgi:2014zwa}, where the parameter $R$ controls the size of the resulting jet region.
The EMD approach yields a natural way to extend the jet maximization algorithm to $N > 1$, and also allows for an alternative definition of $N$-jettiness from \Eq{eq:Njettiness_with_beam_asEMD} based on time-like axes.

As another example, the analysis of recombination schemes in \Sec{sec:seqrec} involved minimizing the transportation cost to merge two particles into one.
For $\beta = 2$, this was equivalent to the $E$-scheme, namely $\kappa = 1$ in \Eq{eq:recomb}, up to the subtlety noted in footnote~\ref{footnote:Escheme}.
Using the massive particle distance, the merged particle in the $E$-scheme has the energy and direction:
\begin{equation}
E_c = E_i + E_j, \qquad n^\mu_c = \frac{E_i n_i^\mu + E_j n_j^\mu}{E_i + E_j},
\end{equation}
where appropriate $T$ subscripts should be included in the hadronic case.
The combined four-vector is
\begin{equation}
p^\mu_c = E_c n^\mu_c = p_i^\mu + p_j^\mu,
\end{equation}
which is a valid expression in both the $e^+e^-$ and hadronic cases.
Thus, the combined four-vector is just the sum of the two particles, which is indeed the desired $E$-scheme behavior.
Note, however, that the interpretation of the jet radius is very different if one uses the massive particle distance, since clustering happens in velocity space.
We leave further studies of the massive particle distance to future work.

\bibliography{eventgeometry}

\end{document}